\documentclass[11pt]{article}
\usepackage[utf8]{inputenc}
\usepackage{fullpage,times}

\usepackage{fullpage}
\usepackage{amsmath,amsfonts,amsthm,amssymb,multirow}
\usepackage{times}
\usepackage{graphicx}
\usepackage{floatpag}
\usepackage{algorithm}
\usepackage[noend]{algpseudocode}
\usepackage{bbold}
\usepackage{enumitem}
\usepackage{subcaption} 
\usepackage{calc}
\usepackage{tikz}
\usetikzlibrary{decorations.markings}
\tikzstyle{vertex}=[circle, draw, inner sep=0pt, minimum size=4pt, fill = black]
\newcommand{\vertex}{\node[vertex]}
\usepackage{graphicx}
\usepackage{tabularx}
\makeatletter
\newcommand{\multiline}[1]{%
  \begin{tabularx}{\dimexpr\linewidth-\ALG@thistlm}[t]{@{}X@{}}
    #1
  \end{tabularx}
}
\makeatother

\makeatletter
\def\BState{\State\hskip-\ALG@thistlm}
\makeatother

\newcommand{\ceil}[1]{\lceil #1 \rceil}
\newcommand{\floor}[1]{\lfloor #1 \rfloor}

\usepackage[compact]{titlesec}
\titlespacing{\section}{0pt}{3ex}{2ex}
\titlespacing{\subsection}{0pt}{2ex}{1ex}
\titlespacing{\subsubsection}{0pt}{0.5ex}{0ex}

\newtheorem{fact}{Fact}[section]

\newtheorem{theorem}{Theorem}[section]

\newtheorem{corollary}{Corollary}[section]
\newenvironment{proofof}[1]{{\bf Proof of #1.  }}{\hfill$\Box$}
\newtheorem{definition}{Definition}[section]
\newtheorem{lemma}{Lemma}[section]

\newtheorem{conjecture}{Conjecture}

\makeatletter
\let\c@fconjecture\c@conjecture
\makeatother

\makeatletter
\let\c@fconj\c@conj
\makeatother

\def \Z {{\mathbb Z}}
\newcommand{\ignore}[1]{}

\title{Graph pattern detection:\\ 
Hardness for all induced patterns and faster non-induced cycles}
\author{Mina Dalirrooyfard, Thuy Duong Vuong and Virginia Vassilevska Williams\thanks{minad@mit.edu, dvuong@mit.edu, virgi@mit.edu, MIT EECS and CSAIL}}
\date{}

\begin{document}

\maketitle
\thispagestyle{empty}
\begin{abstract}
We consider the pattern detection problem in graphs: given a constant size pattern graph $H$ and a host graph $G$, determine whether $G$ contains a subgraph isomorphic to $H$. We present the following new improved upper and lower bounds:
\begin{itemize}
\item We prove that if a pattern $H$ contains a $k$-clique subgraph, then detecting whether an $n$ node host graph contains a {\em not necessarily induced} copy of $H$ requires at least the time for detecting whether an $n$ node graph contains a $k$-clique. The previous result of this nature required that $H$ contains a $k$-clique which is disjoint from all other $k$-cliques of $H$.

\item We show that if the famous Hadwiger conjecture from graph theory is true, then detecting whether an $n$ node host graph contains a {\em not necessarily induced} copy of a pattern with chromatic number $t$ requires at least the time for detecting whether an $n$ node graph contains a $t$-clique. This implies that: (1) under Hadwiger's conjecture for {\em every} $k$-node pattern $H$, finding an {\em induced} copy of $H$ requires at least the time of $\sqrt k$-clique detection and size $\omega(n^{\sqrt{k}/4})$ for any constant depth circuit, and (2) unconditionally, detecting an {\em induced} copy of a random $G(k,p)$ pattern w.h.p. requires at least the time of $\Theta(k/\log k)$-clique detection, and hence also at least size $n^{\Omega(k/\log k)}$ for circuits of constant depth. 


\item We show that for every $k$, there exists a $k$-node pattern that contains a $k-1$-clique and that can be detected as an {\em induced} subgraph in $n$ node graphs in the best known running time for $k-1$-Clique detection. Previously such a result was only known for infinitely many $k$.


\item Finally, we consider the case when the pattern is a directed cycle on $k$ nodes, and we would like to detect whether a directed $m$-edge graph $G$ contains a $k$-Cycle as a {\em not necessarily induced} subgraph. We resolve a {\em 14 year old conjecture} of [Yuster-Zwick SODA'04] on the complexity of $k$-Cycle detection by giving a tight analysis of their $k$-Cycle algorithm. 
Our analysis improves the best bounds for $k$-Cycle detection in directed graphs, for all $k>5$.
\end{itemize}
\end{abstract}
\newpage
\setcounter{page}{1}
\section{Introduction}
One of the most fundamental graph algorithmic problems is Subgraph Isomorphism: given two graphs $G=(V,E)$ and $H=(V_H,E_H)$, determine whether $G$ contains a subgraph isomorphic to $H$. While the general problem is NP-complete, many applications (e.g. from biology \cite{bioapp1,bioapp2}) only need algorithms for the special case in which $H$ is a small graph pattern, of constant size $k$, while the host graph $G$ is large. This graph pattern detection problem is easily in polynomial time: if $G$ has $n$ vertices, the brute-force algorithm solves the problem in $O(n^k)$ time, for any $H$. 

Two versions of the Subgraph Isomorphism problems are typically considered. The first is the {\em induced} version in which one seeks an injective mapping $f:V_H\mapsto V$ so that $(u,v)\in E_H$ if and only if $(f(u),f(v))\in E$. The second is the {\em not necessarily induced} version where one seeks an injective mapping $f:V_H\mapsto V$ so that if $(u,v)\in E_H$ then $(f(u),f(v))\in E$ (however, if $(u,v)\notin E_H$, $(f(u),f(v))$ may or may not be an edge). It is not hard to show (e.g. via color-coding) that when $k$ is a constant, any algorithm for the induced version can be used to solve the not necessarily induced one (for the same pattern) in asymptotically the same time, up to logarithmic factors.

This paper considers two settings of the graph pattern detection problem: (1) Finding {\em induced} patterns of constant size $k$ in dense $n$-node undirected graphs, where the runtime is measured as a function of $n$, and (2) Finding not-necessarily induced patterns in sparse $m$-edge directed graphs; here we focus on $k$-Cycle patterns, a well-studied and important case.

\subsection{Hardness}
A standard generalization of a result of Ne\v{s}etril and Poljak~\cite{Clique2} shows that the induced subgraph isomorphism problem for any $k$-node pattern $H$ in an $n$-node host graph can be reduced in $O(k^2n^2)$ time to the $k$-Clique (or induced $k$-Independent Set (IS)) detection problem in $kn$-node graphs. Thus, for constant $k$, $k$-Clique and $k$-IS are the hardest patterns to detect. 

Following Itai and Rodeh~\cite{Clique1}, Ne\v{s}etril and Poljak~\cite{Clique2} showed that a $k$-Clique (and hence any induced or not-necessarily induced $k$-node pattern) can be detected in an $n$ node graph $G$ asymptotically in time $C(n,k):=M(n^{\lfloor k/3\rfloor}, n^{\lceil k/3\rceil}, n^{\lceil (k-1)/3\rceil})$, where $M(a,b,c)$ is the fastest known runtime for multiplying an $a\times b$ by a $b\times c$ matrix. A simple bound for $M(a,b,c)$ is $M(a,b,c)\leq abc/\min\{a,b,c\}^{3-\omega}$ where $\omega<2.373$ is the exponent of square matrix multiplication \cite{vstoc12,legallmult}, but faster algorithms are known (e.g. Le Gall and Urrutia~\cite{legall-rect-new}). In particular, $C(n,k)\leq O(n^{\omega k/3})$ when $k$ is divisible by $3$.

The $C(n,k)$ runtime for $k$-Clique detection has had no improvements in more than 40 years. Because of this, several papers have hypothesized that the runtime might be optimal for $k$-Cliques (and $k$-Independent Sets) (e.g. \cite{valiantparser,BringmannW17,Lincoln18-cycle}).

Meanwhile, for some $k$-node patterns $H$ that are not Cliques or Independent Sets, specialized algorithms have been developed that are faster than the $C(n,k)$ runtime for $k$-Clique. For instance, 
if $H$ is a $3$-node pattern that is not a triangle or an independent set, it can be detected in $G$ in linear time, much faster than the $C(n,3)=O(n^\omega)$ time for $3$-Clique/triangle. 
Following work of \cite{corneil,Ol90,EiGr04,KlKrMu00,KowalukLL13},
Vassilevska W. et al.~\cite{four-nodes} showed that every $4$-node pattern except for the $4$-Clique and $4$-Independent Set can be detected in $C(n,3)=O(n^\omega)$ time, much faster than the $C(n,4)$ runtime for $4$-Clique. Blaeser et al.~\cite{blaserpattern} recently showed that for $k\leq 8$ there are faster than $C(n,k)$ time algorithms for all non-clique non-independent set $k$-node patterns; for $k\leq 6$, their runtime is $C(n,k-1)$. Independently, we were able to show the same result, using an approach generalizing ideas from \cite{four-nodes}, see the Appendix. 

A natural conjecture, consistent with the prior work so far is that for every $k$ and every $k$-node pattern $H$ that is not a clique or independent set, one can detect it in an $n$ node graph in time $C(n,k-1)$. Blaeser et al. showed that for all $k$ of the form $3\cdot 2^\ell$ for integer $\ell$, there is a $k$-node pattern that (1) is at least as hard to detect as $k-1$-Clique and (2) can be detected in $C(n,k-1)$ time. 
We show that such a pattern exists for {\em all} $k\geq 3$ (Theorem \ref{thm:easierpatterns}).

While there exist $k$-node patterns that can be detected faster than $k$-Clique, it seems unclear how hard $k$-node pattern detection actually is. For instance, it could be that for every $k$, there is {\em some} induced pattern on $k$-nodes that can be detected in say $n^{\log\log(k)}$ time, or even $f(k) n^c$ time, where $c$ is independent of $k$. A Ramsey theoretic result tells us that every $k$-node $H$ either contains an $\Omega(\log k)$ size clique or an $\Omega(\log k)$ size independent set. Hence intuitively, detecting any $k$-node $H$ in an $n$ node graph should be at least as hard as detecting an $\Omega(\log k)$ size clique in an $n$ node graph. The widely believed Exponential Time Hypothesis (ETH) \cite{eth} is known to imply that $k$-Clique cannot be solved in $n^{o(k)}$ time \cite{ChenCFHJKX05}. Coupled with the Ramsey result, ETH should intuitively imply that no matter which $k$-node $H$ we pick, $H$-pattern detection cannot be solved in $n^{o(\log k)}$ time.

 Unfortunately, however, {\em it is still open whether every pattern that contains a $t$-clique is as hard to detect as a $t$-clique} (see e.g. \cite{blaserpattern}\footnote{Blaser et al.~\cite{blaserpattern} show that for the particular types of algorithms that they use a pattern that contains a $k$-clique cannot be found faster than a $k$-clique, and they note that such a result is not known for arbitrary algorithms.}).  In general, it is not clear what makes patterns hard to detect.
 
One of the few results related to this is by Floderus et al.~\cite{FloderusKLL15lower} who showed that if a pattern $H$ contains a $t$-Clique that is disjoint from all other $t$-Cliques in $H$, then $H$ is at least as hard to detect as a $t$-Clique. This implied strong clique-based hardness results for induced $k$-path and $k$-cycle. However, the reduction of \cite{FloderusKLL15lower} fails for patterns whose $k$-Cliques intersect non-trivially. 

The main difficulty in reducing $k$-Clique to the detection problem for other graph patterns $H$ can be seen in the following natural attempt used e.g. by \cite{FloderusKLL15lower}. Say $H$ has a $k$-clique $K$ and let $H'$ be the graph induced by the vertices of $H$ not in $K$. Let $G=(V,E)$ be an instance of $k$-Clique. We'll start by creating $k$ copies of $V$, $V_1,\ldots,V_k$. For every edge $(u,v)$ of $G$, add an edge between the copies of $u$ and $v$ in different parts. Every $k$-clique $C$ of $G$ appears in the new graph $k!$ times; we'll say that the main copy $\bar{C}$ of $C$ has the $i$th vertex of $C$ (in lexicographic order say) appearing in $V_i$. 
Now, add a copy $\bar{H'}$ of $H'$, using fresh vertices, and for every edge $(h,i)$ of $H$ with $h\in H'$ and $i\in K$, add edges from $h\in \bar{H'}$ to all vertices in $V_i$. This forms the new graph $G'$ and guarantees that if $G$ has a $k$-clique $C$, $G'$ contains a copy of $H$ which is just $\bar{C}$ together with $\bar{H'}$. The other direction of the reduction fails miserably however. If $G'$ happens to have a copy of $H$, there is no guarantee that any of the $k$-cliques of $H$ would have a node from each $V_i$ and hence form a clique of $G$. As a simple counterexample (Figure~\ref{fig:counter}) consider $H$ as a $4$-Cycle $(1,2,3,4)$ together with a node $5$ that has edges to all nodes of the $4$-Cycle. Starting from a graph $G$, WLOG we would pick $K$ to be $(1,2,5)$ and $H'=3,4$ and form $G'$ as described. Let $\bar{H'}$ contain the nodes $\bar{3},\bar{4}$ and let the parts of $G$ be $V_1,V_2,V_5$. Now the reduction graph $G'$ might contain a copy of $H$ even if $G$ has no $3$-cliques, as $\bar{4}$ could represent $5$, and $1,3$ and $2,4$ could be represented by two nodes each in $V_1$ and $V_5$ respectively; see Figure~\ref{fig:counter}. Hence the copy of $H$ wouldn't use $V_2$ at all and doesn't represent a triangle in $G$.

\begin{figure}
\centering
    \includegraphics[width=0.5\textwidth]{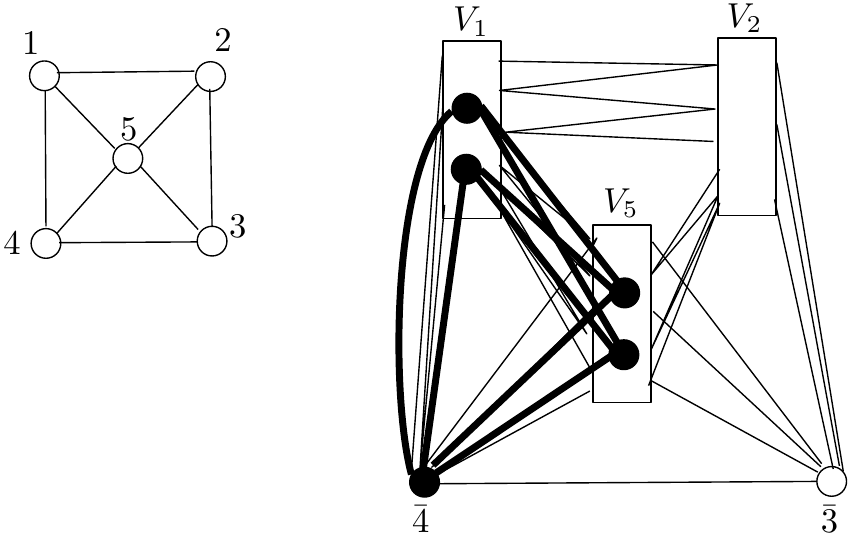}
    \caption{An example of how a simple reduction attempt fails to reduce $3$-Clique to $H$. The edges between the $V_i$ are determined by the $3$-Clique instance.}
	\label{fig:counter}
\end{figure}

One could try to modify the reduction, say by representing the nodes of $H'$ by copies of the vertices of $G$, as with $K$. However, the same issues arise, and they seem to persist in most natural reduction attempts.

With an intricate construction, we show how to overcome this difficulty.  
Our first main theorem is that patterns that contain $t$-cliques are indeed at least as hard as $t$-Clique, and in fact we prove it for the {\em not necessarily induced} case which automatically gives a lower bound for the induced case (Theorem~\ref{thm:K_tLB}
 in the body):
\begin{theorem}
Let $G=(V,E)$ be an $n$-node, $m$-edge graph and let $H$ be a $k$-node pattern such that $H$ has a $t$-clique as a subgraph. Then one can construct a new graph $G^*$ of at most $nk$ vertices in $O(k^2m+kn)$ time such that $G^*$ has a not necessarily induced subgraph isomorphic to $H$ if and only if $G$ has a $t$-clique. 
\end{theorem}

Note that since the not necessarily induced pattern detection can be solved with the induced version, a lower bound for the not necessarily induced pattern detection gives a lower bound for the induced version. 
Since for every $k$-node graph $H$, either $H$ or its complement contains a clique of size $\Omega(\log k)$, ETH implies that no matter which $k$-node $H$ we pick, induced $H$-pattern detection cannot be solved in $n^{o(\log k)}$ time. 

Our second theorem shows that some patterns are even harder, as in fact {\em the hardness of a pattern grows with its chromatic number!} 

Our theorem relies on the widely believed Hadwiger conjecture~\cite{hadwiger} from graph theory which roughly states that every graph with chromatic number $t$ contains a $t$-clique as a minor. The Hadwiger conjecture is known to hold for $t\leq 6$ \cite{hadwiger6} and to almost hold for $t=7$ \cite{KawarabayashiT05} (It is equivalent to the $4$-Color Theorem for $t=5,6$ \cite{hadwiger6,fourcolorequiv,fourcolor2}.). It also holds for almost all graphs \cite{hadwigerrandom}. 
Our lower bound theorem, which also proved for the {\em not necessarily induced} case (Theorem~\ref{thm:chromaticLB} in the body) is:
\begin{theorem}
Let $G=(V,E)$ be an $n$-node graph and let $H$ be a $k$-node pattern with chromatic number $t$, for $t>1$. Then assuming that Hadwiger conjecture is true, one can construct $G^*$ on at most $nk$ vertices in $O(n^2k^2)$ time such that $G^*$ has a not necessarily induced subgraph isomorphic to $H$ if and only if $G$ has a $t$-clique. 
\end{theorem}

This is the first connection between the Hadwiger conjecture and Subgraph Isomorphism, to our knowledge.
Let us see some exciting consequences of this theorem. First, we get that if $t$ is the maximum of the chromatic numbers of $H$ and its complement, then an induced $H$ is at least as hard as $t$-Clique to detect. Now, it is a simple exercise that the maximum of the chromatic number of a $k$-node graph and its complement is at least $\sqrt k$. 
{\em Thus, every induced $H$ on $k$-nodes is at least as hard as $\sqrt k$-Clique. There are no easy induced patterns.}

\begin{corollary} No matter what $k$-node $H$ we take, under ETH and the Hadwiger Conjecture, the induced subgraph isomorphism problem for $H$ in $n$-node graphs cannot be solved in $n^{o(\sqrt k)}$ time.
\end{corollary}
This is the first result of such generality.

A second consequence comes from circuit complexity. Rossman~\cite{RossmanClique} showed that for any constant integers $k$ and $d$, any circuit of depth $d$ requires size $\omega(n^{k/4})$ to detect a $k$-Clique. Because of the simplicity of our reduction (it can be implemented in constant depth), we also obtain a circuit lower bound for induced pattern detection for any $H$ node subgraph:

\begin{corollary} Let $d$ and $k$ be any integer constants. No matter what $k$-node $H$ we take, under the Hadwiger Conjecture, any depth $d$ circuit for
the induced subgraph isomorphism problem for $H$ in $n$-node graphs requires size $\omega(n^{\sqrt k / 4})$.
\end{corollary}

A third consequence is that in fact almost all $k$-node induced patterns are very hard -- at least as hard as $\Theta(k/\log k)$-Clique. Consider an Erd\"{o}s-Renyi graph $H$ from $G(k,p)$ for constant $p$. It is known \cite{hadwigerrandom} that the Hadwiger conjecture holds for $H$ with high probability. Moreover, the chromatic number of such graphs (and their complements) is with high probability $\Theta(k/\log k)$ \cite{chromaticrandom}; meanwhile the clique and independent set size is only $O(\log k)$. Thus our chromatic number theorem significantly strengthens our first theorem.

\begin{corollary} For almost all $k$-node patterns $H$, under ETH, induced $H$ detection in $n$ node graphs cannot be done in $n^{o(k/\log k)}$ time.\end{corollary}

We also immediately obtain, via Rossman's lower bound, that for almost all $k$-node patterns $H$, any constant depth circuit that can detect an induced $H$ requires size $n^{\Omega(k/\log k)}$.

\subsection{Detecting not-necessarily induced directed $k$-Cycles.}
Some of the most striking difference between the complexity of induced and not-necessarily induced subgraph detection is in the $k$-Path and $k$-Cycle problems. Since a $k$-Path and a $k$-Cycle both contain an independent set on $\lfloor k/2\rfloor$ nodes, the induced version of their subgraph detection problems is at least as hard as detecting $\lfloor k/2\rfloor$-cliques, and 
needs $C(n,\lfloor k/2\rfloor)$ time unless there is a breakthrough in Clique detection. Thus also under ETH, induced $k$-Path and $k$-Cycle cannot be solved in $n^{o(k)}$ time. Monien~\cite{monien}, however, showed that for all constants $k$, a non-induced $k$-Path can be detected with constant probability in {\em linear} time. Thus, for constant $k$, the non-induced $k$-Path problem has an essentially optimal (randomized) algorithm. 
With the same ideas, a $k$-Cycle can be found in $\tilde{O}(n^\omega)$ time. Due to the tight relationship between triangle detection and Boolean matrix multiplication (e.g. \cite{focsy}), this runtime is often conjectured to be optimal for dense graphs. For sparse graphs, however, there has been a lot of active research in improving the runtime of $k$-Cycle detection, and it is completely unclear what the best runtime should be.

Alon, Yuster and Zwick~\cite{AlYuZw97} gave several algorithms for both directed and undirected cycle detection. The bounds for directed graphs are as follows. For $3$-Cycles (triangles) \cite{AlYuZw97} gives an algorithm running in time $O(m^{2\omega/(\omega+1)})\leq O(m^{1.41})$, which is still the fastest algorithm for the problem in sparse graphs. For general $k$, one can find a $k$-Cycle in time $O(m^{2-2/k})$ if $k$ is even and in time $O(m^{2-2/(k+1)})$ if $k$ is odd. These last algorithms do not use matrix multiplication.

Yuster and Zwick~\cite{YuZw04} set out to improve upon the general $k$-Cycle algorithms above using fast matrix multiplication. They presented an algorithm that combines most known techniques for cycle detection and works for arbitrary $k\geq 3$. 
However, they were not able to analyze the complexity of their algorithm in general. They showed that for $k=4$, the algorithm runs in $O(m^{(4\omega-1)/(2\omega+1)})\leq O(m^{1.48})$ time, and that for $k=5$, it runs in time $O(m^{3\omega/(\omega+2)})\leq O(m^{1.63})$. Both bounds improve the runtimes from \cite{AlYuZw97}.

Already for $k=6$ the analysis of the algorithm seemed very difficult. Yuster and Zwick ran computer simulations to find the worst case runtime for $k=6$ and beyond and came up with conjectures for what the runtime should be for $k=6$ and for all odd $k$. They also stated that for even $k$ larger than $6$, the runtime expression is likely extremely complicated. Their conjectures have remained unproven for over 14 years.

In this paper we present an analysis of the running time of the Yuster-Zwick algorithm, proving the two conjectures (for $k=6$ and odd $k$). We give an analysis of the runtime for even $k$ as well. Our bound is tight, assuming that the matrix multiplication exponent $\omega$ is $2$. For larger values of $\omega$, the tight bound on the runtime is a step function of $\omega$ that remains to be analyzed. Our final result is as follows:

\begin{theorem}\label{thm:cycles} There is an algorithm for $k$-Cycle detection in $m$ edge directed graphs (the Yuster-Zwick algorithm) which runs in $\tilde{\Theta}(m^{c_k})$ time, where
\begin{itemize}
\item $c_k=\omega(k+1)/(2\omega+k-1)$ when $k$ is odd, 
\item $c_4=(4\omega-1)/(2\omega+1)$
\item \begin{equation*}
c_6 = \begin{cases} 
\frac{10 \omega -3}{4 \omega +3}, & \text{if}\ 2 \leq \omega \leq \frac{13}{6}\\
\frac{22-4\omega}{17 - 4\omega}, & \text{if}\ \frac{13}{6} \leq \omega \leq \frac{9}{4}\\
\frac{11 \omega -2}{4\omega + 5}, & \textit{if}\ \frac{9}{4} \leq \omega \leq \frac{16}{7} \\
\frac{10 - \omega}{7 - \omega}, & \textit{if}\ \frac{16}{7} \leq \omega \leq \frac{5}{2}
\end{cases}
\end{equation*}
\item $c_k\leq (k\omega-4/k)/(2\omega+k-2-4/k)$ for all even $k\geq 4$. This is tight for $\omega=2$.
\end{itemize}
\end{theorem}

\paragraph{Related Work.}
 Vassilevska~\cite{thesis} showed that $K_k-e$ (a $k$-clique missing an edge) can be found in $O(n^{k-1})$ time without using fast matrix multiplication, whereas the fastest algorithms for $k$-Clique without fast matrix multiplication run in $O(n^k/\log^{k-1} n)$ time~\cite{Vassilevska09}; this was recently improved by Blaeser et al. \cite{blaserpattern} who showed that every $k$ node pattern except the $k$-Clique and $k$-Independent Set can be detected in time $O(n^{k-1})$.  Before this, Floderus et al. \cite{FloderusKLL13small} showed that 5 node patterns\footnote{All patterns except for $K_5$, $K_4+e$, $(3,2)-$fan, gem, house, butterfly, bull, $C_5$, $K_{1,4}$, $K_{2,3}$ and their complements; for these subgraphs the fastest runtime remained $C(n,5)\leq O(n^{4.09})$.} can be found in $O(n^4)$ time, again without using fast matrix multiplication.

Some other related work includes improved algorithms for subgraph detection when $G$ has special structure (e.g. \cite{KowalukL17walcom} and \cite{FominLRSR12}).
Other work counts the number of occurrences of a pattern in a host graph (e.g. \cite{KowalukLL13,VW09,CurticapeanDM17}). Finally, there is some work on establishing conditional lower bounds. Floderus et al.~\cite{FloderusKLL15lower} produced reductions from $k$-Clique (or $k$-Independent Set) to the detection problem of $\ell$-patterns for $\ell>k$ (but still linear in $k$). They show for instance that finding an induced $k$-path is at least as hard as finding an induced $k/2$-independent set. Lincoln et al. \cite{Lincoln18-cycle} give conditional lower bounds for not-necessarily induced directed $k$-cycle detection.
For instance, they show that if $k$-Clique requires essentially $C(n,k)$ time, then finding a directed $k$-Cycle in an $m$ edge graph requires $m^{2\omega k/(3(k+1))-o(1)}$ time. This lower bound is lower than the upper bounds in this paper, but they do show that superlinear time is likely needed.

Detecting $k$-Cycles in undirected graphs is an easier problem, when $k$ is an even constant. Yuster and Zwick~\cite{evencycles} showed that a $k$-Cycle in an undirected graph can be detected (and found) in $O(n^2)$ time for all even constants $k$. Dahlgaard et al. \cite{DahlgaardKS17} extended this result showing that $k$-Cycles for even $k$ in $m$-edge graphs can be found in time $\tilde{O}(m^{2k/(k+1)})$. Their result implies that of \cite{evencycles}, as by a result of Bondy and Simonovits \cite{BoSi74}, any $n$ node graph with $\geq 100k n^{1+1/k}$ edges must contain a $2k$-Cycle.
When $k$ is an odd constant, the $k$-Cycle problems in undirected and directed graphs are equivalent (see e.g. \cite{thesis}).


\section{Lower bounds}
\label{induced-LB}
In this section we consider the problem of detecting and finding a (not necessarily induced) copy of a given small pattern graph $H$ in a host graph $G$. This is the variant of subgraph isomorphism in which the pattern $H$ is fixed, on a constant $k$ number of vertices, and $G=(V,E)$ with $|V|=n$ is given as an input. We focus on the hardness of this problem: we show that any fixed pattern that has a $t$-clique as a subgraph, is not easier to detect as a subgraph than a $t$-clique, formally stated as Theorem \ref{thm:K_tLB}.
First, we start by an easier case of the theorem where the pattern is $t$-chromatic to depict the main idea of our proof and then we proceed with the proof of the theorem for all patterns. Recall that a \emph{proper vertex coloring} of a graph is an assignment of colors to each of its vertices such that no edge connects two identically colored vertices. If the set of colors is of size $c$, we say that the graph is $c$-colorable. The \emph{chromatic number} of a graph is the smallest number $c$ for which the graph is $c$-colorable, and 
we call such graph \emph{$c$-chromatic}. In the second part of this section, we prove a stronger lower bound using Hadwiger conjecture, showing that under this conjecture any $t$-chromatic pattern is not easier to detect as a subgraph than a $t$-clique.

\begin{theorem}
\label{thm:K_tLB}
Let $G=(V,E)$ be an $n$-node graph and let $H$ be a $k$-node pattern such that $H$ has a $t$-clique as a subgraph. Then one can construct $G^*$ on at most $nk$ vertices in $O(k^2m+kn)$ time such that $G^*$ has a (not necessarily induced) subgraph isomorphic to $H$ if and only if $G$ has a $t$-clique. 
\end{theorem}

\subsection{Simple case: $t$-Chromatic patterns}
We show Theorem \ref{thm:K_tLB} when $H$ is $t$-chromatic in addition to having a $t$-clique as a subgraph. Construct $G^*$ as follows: For each $v\in H$, let $G_v$ be a copy of the vertices of $G$ as an independent set. For any two vertices $v$ and $u$ in $H$ where $vu$ is an edge, add the following edges between $G_v$ and $G_u$: for each $w_1$ and $w_2$ in $G$, add an edge between the copy of $w_1$ in $G_v$ and the copy of $w_2$ in $G_u$ if and only if $w_1w_2$ is an edge in $G$. So $G^*$ has $nk$ vertices and since for each pair of vertices $u,v\in H$ we have at most $m$ edges between $G_u$ and $G_v$, the construction time is at most $O(k^2m+kn)$.

Now we show that $G$ has a $t$-clique as a subgraph if and only if $G^*$ has $H$ as a subgraph. First suppose that $G$ has a $t$-clique, say $T=v_1,\ldots,v_t$. Consider a $t$-coloring of the vertices of $H$, with colors $1,\ldots,t$. For each $w\in H$, pick $v_i$ from $G_w$ if $w$ is of color $i$. Call the induced subgraph on these vertices $H^*$. We show that $H^*$ is isomorphic to $H$: map each $w\in H$ to the vertex picked from $G_w$ in $G^*$. If $w$ and $w'$ are adjacent in $H$, then their colors are different, so the vertices that are picked from $G_w$ and $G_{w'}$ are different vertices of $G$, and they are part of the clique $T$, so they are adjacent. If $w$ and $w'$ are not adjacent, we don't have any edges between $G_w$ and $G_{w'}$, so the vertices picked from them are not adjacent. 

For the other direction, we show that if $G^*$ has $H$ as a subgraph then $G$ has a $t$-clique. Since $H$ has a $t$-clique as a subgraph, $G^*$ also has a $t$-clique as a subgraph. Suppose the vertices of this clique are $W=\{w_1,\ldots, w_t\}$ where $w_i$ is a copy of $v_i\in G$. Each pair of vertices of the clique are in different copies of $G$, as these copies are independent sets. Moreover, for each $i,j\in\{1,\ldots,t\}$, since $w_i$ and $w_j$ are adjacent, they correspond to different vertices in $G$, so $v_i\neq v_j$. Since we connect two vertices in $G^*$ if their corresponding vertices in $G$ are connected, this means that $v_i$ and $v_j$ are connected in $G$. So $v_1,\ldots,v_t$ form a $t$-clique in $G$. 

\subsection{General case}
Define a {\em $t$-clique covering} of a pattern $H$ to be a collection $\mathcal{C}$ of sets of vertices of $H$, such that the induced subgraph on each set is $t$-colorable, and for any $t$-clique $T$ of $H$, there is a set in $\mathcal{C}$ that contains all the vertices of $T$. For example, in Figure \ref{fig:Gstarex}, the graph $H_{ex}$ has the following $3$-clique covering of size $2$: $\{\{a_1,a_2,a_3,a_6\},\{a_3,a_4,a_5,a_1,a_6\}\}$.

For each $H$ we have at least one $t$-clique covering by considering the vertices of each $t$-clique of $H$ as one set. However we are interested in the smallest collection $\mathcal{C}$. So for a fixed $t$, we define $p(H)$ to be the smallest integer $r\ge 1$, such that there is a $t$-clique covering of $H$ of size $r$. We call a $t$-clique covering of size $p(H)$ a {\em minimum} $t$-clique covering. For example, if $H$ is $t$-colorable, $p(H)=1$ as the whole vertex set is the only set that the $t$-clique covering has. If $H$ is not $t$-colorable but has a $t$-clique, then $p(H)>1$. 

\begin{figure}
  \centering
    \includegraphics[width=0.6\textwidth]{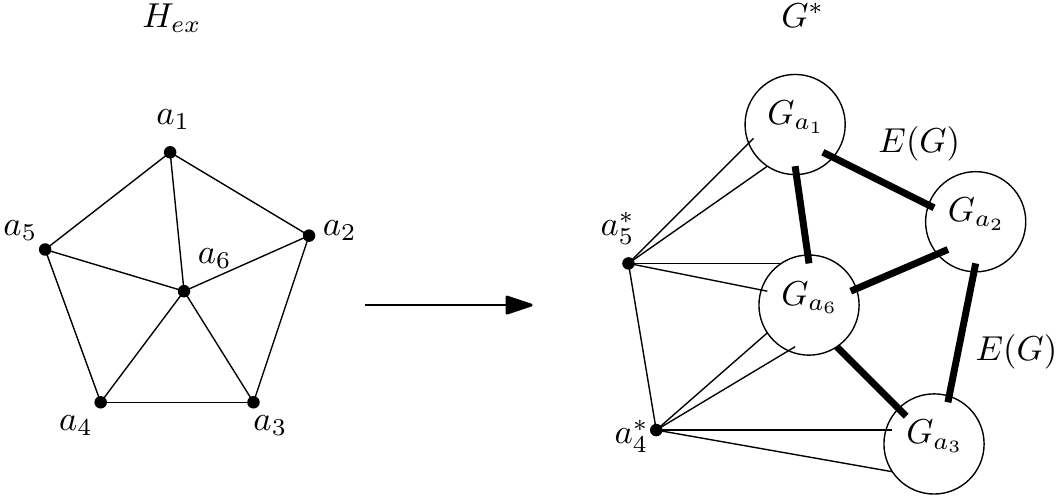}
    \caption{Graph $H_{ex}$ on the left. The largest clique of this graph is a triangle. $H_{ex}$ is $4$-chromatic, so $p(H_{ex})>1$. We have $p(H_{ex})=2$, as a minimum $3$-clique covering for it is $\{\{a_1,a_2,a_3,a_6\},\{a_3,a_4,a_5,a_1,a_6\}\}$. The graph $G^*$ is on the right, thick edges represent the way the edges are specified according to $E(G)$ between two copies of $G$.}
	\label{fig:Gstarex}
\end{figure}

\begin{proofof}{Theorem \ref{thm:K_tLB}}
Let $\mathcal{C}=\{C_1,\ldots,C_r\}$ be a minimum $t$-clique covering of $H$, where $r=p(H)$. 
The vertex set of $G^*$ is the following: For each vertex $v\in C_1$, let $G_v$ be a copy of the vertices of $G$ as an independent set. For each vertex $v\in V(H)\setminus C_1$, let $v^*$ be a copy of $v$ in $G^*$. The edge set of $G^*$ is as follows: For each two vertices $v,u\in C_1$ that $uv$ is an edge in $H$, add the following edges between $G_v$ and $G_u$: for each $w_1$ and $w_2$ in $G$, add an edge between the copy of $w_1$ in $G_v$ and the copy of $w_2$ in $G_u$ if and only if $w_1w_2$ is an edge in $G$. For each two vertices $u\in C_1$ and $v\in V(H)\setminus C_1$ that $uv$ is an edge in $H$, connect $v^*$ to all the vertices in $G_u$. For each two vertices $u,v \in V(H)\setminus C_1$ that $uv$ is an edge in $H$, connect $u^*$ and $v^*$. The way $G^*$ is constructed is shown in Figure \ref{fig:Gstarex} for the particular pattern $H_{ex}$ with maximum clique $3$.

Now we show that $G$ has a $t$-clique as a subgraph if and only if $G^*$ has $H$ as a subgraph. First suppose that $G$ has a $t$-clique, say $v_1,\ldots,v_t$. Consider a $t$-coloring of vertices of $C_1$, with colors $1,\ldots,t$. Let $H^*$ be the subgraph on the following vertices in $G^*$: for each $w\in C_1$, pick $v_i$ from $G_w$ if $w$ is of color $i$. For each $w\in V(H)\setminus C_1$, pick $w^*$. We show that $H$ is isomorphic to $H^*$: for each $w\in C_1$, map $w$ to the vertex picked from $G_w$, and for each $w\in V(H)\setminus C_1$, map $w$ to $w^*$. If $w,u\in C_1$ such that $wu\in E(H)$, then their colors are different in the $t$-coloring of $C_1$, and so the vertices that are picked from $G_w$ and $G_{u}$ are different vertices of $G$ and part of the $t$-clique of $G$, so they are attached. If $wu$ is not an edge, then there is no edge between $G_w$ and $G_{u}$. If $w\in C_1$ and $u\in V(H)\setminus C_1$ and $wu$ is an edge in $H$, then $u^*$ is attached to all vertices in $G_{w}$ including the vertex that is picked from $G_{w}$ for $H^*$. If $wu$ is not an edge, then there is no edge between $u^*$ and $G_w$. If $w,u\in V(H)\setminus C_1$, then $u^*,w^*$ are both picked in $H^*$ and they are adjacent in $G^*$ if and only if $w$ and $u$ are adjacent in $H$. 

For the other direction, we show that if $G^*$ has a subgraph $H^*$ isomorphic to $H$, then $G$ has a $t$-clique. Let $S_1=\cup_{v\in C_1} G_v$. First suppose that $H^*$ has a $t$-clique $T$ using vertices in $S_1$. Since for each $v\in C_1$, $G_v$ is an independent set, no two vertices of $T$ are in the same $G_v$. So there are $t$ vertices of $H$, $v_1,\ldots,v_t$ such that $T$ has a vertex in each $G_{v_i}$. Let this vertex be a copy of $w_i\in G$. Since for each $i,j\in\{1,\ldots,t\}$, $i\neq j$, the copies of $w_i$ and $w_j$ are adjacent in $G^*$, we have that $w_i\neq w_j$ and they are adjacent in $G$. So $\{w_1,\ldots,w_t\}$ form a $t$-clique in $G$. 

So assume that the induced subgraph on $V(H^*)\cap S_1$ in $G^*$ has no clique. As $S_1$ has all the vertices in $G^*$ that correspond to the vertices in $C_1$, we define similar sets for other $C_i$s. For $i\in \{2,\ldots,p(H)\}$, let $S_i'=\cup_{v\in C_i\cap C_1} G_v$, $S_i''=\cup_{v\in C_i\setminus C_1} v^*$ and $S_i=S_i'\cup S_i''$. First note that the induced subgraph on $S_i$ is $t$-colorable: Consider the $t$-coloring of $C_i$. For each $v\in C_i\setminus C_1$, color $v^*$ the same as $v$. For each $v\in C_i\cap C_1$, color all vertices in $G_v$ the same as $v$. 

Now we show that any $t$-clique in $H^*$ is in one of the sets $S_2,\ldots,S_{p(H)}$. This means that the collection $\{S_2\cap V(H^*), \ldots, S_{p(H)}\cap V(H^*)\}$ is a $t$-clique covering for $H$ with size $p(H)-1$, which is a contradiction. Consider a $t$-clique $T=v_1,\ldots,v_t$ in $H^*$. Each $v_i$ is in one of the copies of $G$ or is a copy of a vertex in $H$. So for each $v_i$, there is some vertex $w_i\in H$, such that $v_i\in G_{w_i}$ and $w_i\in C_1$ if $v_i\in S_1$, or $v_i=w_i^*$ and $w_i\notin C_1$ if $v_i\notin S_1$. Since for each $i,j$, $v_i$ and $v_j$ are adjacent in $G^*$, this means that $w_i$ and $w_j$ are different vertices in $H$ and they are adjacent. So $W=\{w_1,\ldots,w_t\}$ form a clique in $H$. Since $T\notin S_1$, wlog we can assume that $v_1\notin S_1$. So $w_1\notin C_1$. So the $t$-clique $W$ is not in $C_1$, and so it is in $C_i$, for some $2\le i\le p(H)$. Hence, $T\in S_i$. 
\end{proofof}

\begin{corollary}
Let $H$ be a $k$-node pattern that has a $t$-clique or a $t$-independent set as a subgraph. Then the problem of finding $H$ as an induced subgraph in an $n$-node graph is at least as hard as finding a $t$-clique in an $O(n)$-node graph.
\end{corollary}

\subsection{A Stronger Lower Bound}
One of the oldest conjectures in graph theory is Hadwiger conjecture which intoduces a certain structure for $t$-chromatic graphs. Assuming that this conjecture is true, we show that any fixed pattern with chromatic number $t$ is not easier to detect as an induced subgraph than a $t$-clique. This strengthens the previous lower bound because the size of the maximum clique of a pattern is at most its chromatic number, and moreover there are graphs with maximum clique of size two but large chromatic number. 

\begin{conjecture}[Hadwiger's Conjecture]
\label{conj:Hadwiger}
Let $H$ be a graph with chromatic number $t$. Then one can find $t$ disjoint connected subgraphs of $H$ such that there is an edge between every pair of subgraphs.
\end{conjecture}

Contracting the edges within each of these subgraphs so that each subgraph collapses to a single vertex produces a $t$-clique as a minor of $H$. This is the property we are going to use to show that $H$ is at least as hard to detect as a $t$-clique. Our main theorem is as follows.

\begin{theorem}
\label{thm:chromaticLB}
Let $G=(V,E)$ be an $n$-node graph and let $H$ be a $k$-node $t$-chromatic pattern, for $t>1$. Then assuming that Hadwiger conjecture is true, one can construct $G^*$ on at most $nk$ vertices in $O(n^2k^2)$ time such that $G^*$ has a (not necessarily induced) subgraph isomorphic to $H$ if and only if $G$ has a $t$-clique. 
\end{theorem}

To prove Theorem \ref{thm:chromaticLB}, we use a similar approach as Theorem \ref{thm:K_tLB}. The approach of Theorem \ref{thm:K_tLB} is covering the maximum cliques of the pattern by a collection of subgraphs. However, since in Theorem \ref{thm:chromaticLB} the pattern doesn't necessarily have a $t$-clique, we cover another particular subgraph of the pattern, and hence we introduce a similar notion as $t$-clique covering for this subgraph. 

Let $F$ be a graph with a vertex (not necessarily proper) coloring $C:V(F)\rightarrow \{1,\ldots,t\}$. We say that $F$ has a $K_t$ minor {\em with respect to the coloring $C$} if the vertices of each color induce a connected subgraph and for every color there is an edge from one of the vertices of that color to one of the vertices of every other color. For example, in Figure \ref{fig:Gstarex}, consider the following coloring for $H_{ex}$: $C_{ex}:\{a_1,\ldots,a_6\}\rightarrow \{1,\ldots,4\}$, where $C_{ex}(a_1)=C_{ex}(a_2)=1$, $C_{ex}(a_3)=C_{ex}(a_4)=2$, $C_{ex}(a_5)=3$ and $C_{ex}(a_6)=4$. Clearly $H_{ex}$ has a $K_4$ minor with respect to the coloring $C_{ex}$.

Let $F$ and $H$ be two fixed graphs, where $F$ is $t$-chromatic. We say that $H$ is {\em $(K_t, F)$ minor colorable} if there is a (not necessarily proper) coloring $C:V(H)\rightarrow \{1,\ldots,t\}$ such that any induced copy of $F$ in $H$ has a $K_t$ minor with respect to $C$. For example, in Figure \ref{fig:hadwigerLB}, the graph $H_{ex}'$ has graph $H_{ex}$ (Figure \ref{fig:Gstarex}) as a $4$-chromatic subgraph, and it is $(K_4,H_{ex})$ minor colorable: There are exactly two copies of $H_{ex}$ in $H_{ex}'$, one with vertex set $\{a_1,\ldots, a_6\}$ and one with vertex set $\{a_1,a_4,a_5,a_6,a_7,a_8\}$, and both have a $K_4$ minor with respect to the coloring given in Figure \ref{fig:hadwigerLB}. Note that \emph{minor colorability} and \emph{colorability} are different: recall that $c$-colorability of a graph for an integer $c$ means that the graph has a proper coloring using $c$ colors, and the graph is $c$-chromatic if $c$ is the smallest integer such that the graph is $c$-colorable. 

Let $H$ be a pattern and let $F$ be a $t$-chromatic subgraph of $H$. As a generalization to a $t$-clique covering of $H$, we define an {\em $F$-covering} of $H$ to be a collection $\mathcal{C}$ of sets of vertices of $H$, such that the induced subgraph of each set is $(K_t, F)$ minor colorable, and each copy of $F$ is completely inside one of the sets in $\mathcal{C}$.

For any graph $H$, we have at least one $F$-covering by considering the vertices of each copy of $F$ as one set where the $(K_t, F)$ minor colorablitiy of each set comes from Conjecture \ref{conj:Hadwiger}. Similar to $t$-clique coverings we are interested in the smallest collection $\mathcal{C}$ among all $F$-coverings. So for a fixed number $t$ and a $t$-chromatic subgraph $F$ of $H$, we define $p_F(H)$ to be the smallest integer $r\ge1$, such that there is an $F$-covering of $H$ of size $r$. We call an $F$-covering of size $p_F(H)$ a {\em minimum} $F$-covering. Note that $p_{K_t}(H)=p(H)$. For example, in Figure \ref{fig:hadwigerLB}, $p_{H_{ex}}(H_{ex}')=1$, according to the coloring given in the figure. 

\begin{figure*}
  \centering
    \includegraphics[width=0.8\textwidth]{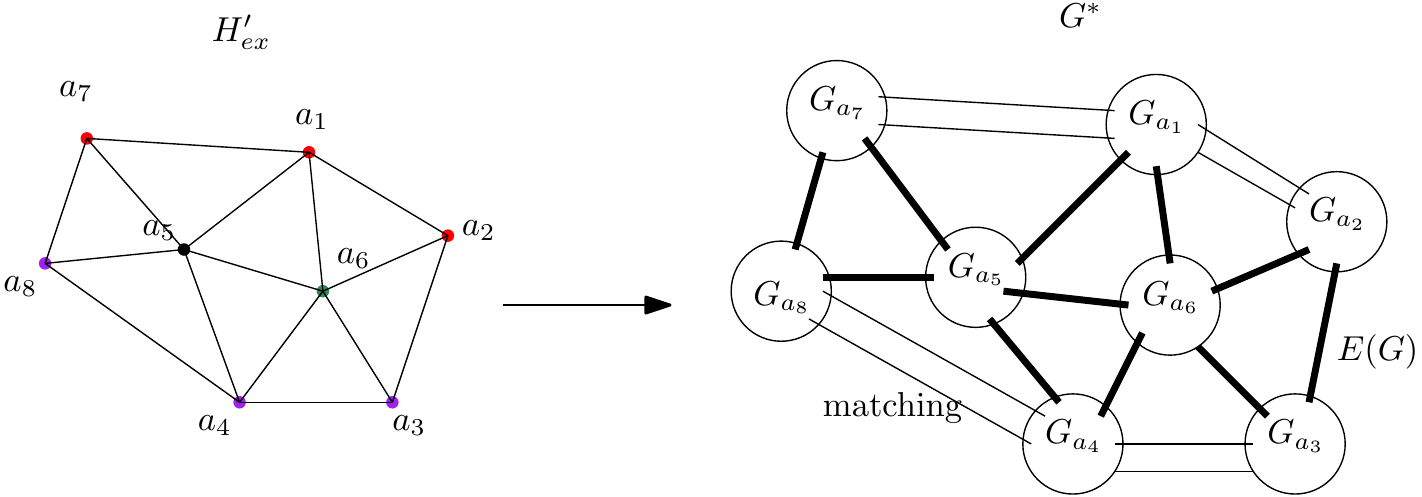}
    \caption{The $4$-chromatic graph $H_{ex}'$ on the left side has the coloring $C_{ex}'$ which makes it $(K_4,H_{ex})$ minor colorable: $C_{ex}'(a_1)=C_{ex}'(a_2)=C_{ex}'(a_7)=1$, $C_{ex}'(a_3)=C_{ex}'(a_4)=C_{ex}'(a_8)=2$, $C_{ex}'(a_5)=3$, $C_{ex}'(a_6)=4$. On the right side we show how $G^*$ is constructed as it is described in the proof of Theorem \ref{thm:chromaticLB}. The double edges indicate a matching where nodes that are copy of the same vertex in $G$ are connected. The thick edges represent the way we add edges according to $E(G)$. }
	\label{fig:hadwigerLB}
\end{figure*}

Now we are ready to prove Theorem \ref{thm:chromaticLB}.

\begin{proofof}{Theorem \ref{thm:chromaticLB}}
We are going to mimic the proof of Theorem \ref{thm:K_tLB}, and so we are going to carefully choose a subgraph $F$ and consider the minimum $F$-covering of it. 

Let $z$ be the largest integer such that every $(z-1)$-node subgraph of $H$ is $t-1$ colorable. Let $F$ be a $t$-chromatic subgraph of $H$ on $z$ nodes with maximum number of edges. Note that $F$ is an induced subgraph of $H$.
In Figure \ref{fig:hadwigerLB}, $H=H_{ex}'$ is $4$-chromatic and one can check that any subgraph on $6$ vertices or less is $3$ colorable. In this graph $z=7$ and $F=H_{ex}$. 

Now suppose that $\mathcal{C}=\{C_1,\ldots,C_{r}\}$ is a minimum $F$-covering of $H$, where $r=p_F(H)$. Let $f:C_1\rightarrow \{1,\ldots,t\}$ be a $(K_t,F)$ minor coloring of $C_1$. Define the vertex set of $G^*$ as follows: For each vertex $v\in C_1$, let $G_v$ be a copy of $G$ as an independent set. For each vertex $v\in V(H)\setminus C_1$, let $v^*$ be a copy of $v$ in $G^*$. The edge set of $G^*$ is as follows: For each pair of vertices $u,v\in C_1$, if $uv$ is not an edge in $H$ we don't add any edges between $G_u$ and $G_v$. If $uv$ is an edge and $f(u)= f(v)$, then add the following edges between $G_u$ and $G_v$: For each $w\in G$, add an edge between the copy of $w$ in $G_u$ and the copy of $w$ in $G_v$ (So we have a complete matching between $G_u$ and $G_v$). If $uv$ is an edge and $f(u)\neq f(v)$, then add the following edges between $G_u$ and $G_v$: for each $w_1$ and $w_2$ in $G$, add an edge between the copy of $w_1$ in $G_u$ and the copy of $w_2$ in $G_v$ if and only if $w_1w_2$ is an edge in $G$. For each pair of vertices $u\in C_1$ and $v\in V(H)\setminus C_1$ such that $uv$ is an edge in $H$, add an edge between  $v^*$ and all vertices in $G_u$. For each pair of vertices $u,v \in V(H)\setminus C_1$ such that $uv$ is an edge in $H$, add an edge between $u^*$ and $v^*$ in $G^*$. In Figure \ref{fig:hadwigerLB}, $H_{ex}'$ has a $H_{ex}$-covering of size $1$ which is the whole graph. On the right side of the figure we show how $G^*$ is constructed.  

Now we show that $G$ has a $t$-clique as a subgraph if and only if $G^*$ has $H$ as a subgraph. First, suppose that $G$ has a $t$-clique, say $T=v_1,\ldots, v_t$. Let $H^*$ be the induced subgrpah on the following vertices in $G^*$: for each $w\in C_1$, pick $v_i$ from $G_w$ if $f(w)=i$. For each $w\in V(H)\setminus C_1$, pick $w^*$. We show that $H$ is isomorphic to $H^*$: for each $w\in C_1$, map $w$ to the vertex picked from $C_w$, and for each $w\in V(H)\setminus C_1$, map $w$ to $w^*$. If $u,w\in C_1$ and they are not adjacent, then there is no edge between $G_{u}$ and $G_{w}$. If $uw$ is an edge in $H$, then if $f(u)=f(w)=i$, we picked $v_i$ from both $G_{u}$ and $G_{w}$ and hence there are adjacent (note that in this case the edges between $G_{u}$ and $G_{w}$ form a complete matching). If $f(u)\neq f(w)$, then the vertices that we picked from $G_{u}$ and $G_{w}$ are copies of different vertices of the clique $T$, and so they are adjacent in $G^*$. If $u\in C_1$ and $w\in V(H)\setminus C_1$ and $uw$ is an edge in $H$, then $w^*$ is attached to all vertices in $G_u$, so it is adjacent to the vertex chosen from $G_u$ for $H^*$. If $uw$ is not an edge, then there is no edge between $w^*$ and $G_u$. If $u,w\in V(H)\setminus C_1$, then $u^*$ and $w^*$ are connected in $H^*$ if and only if $uw$ are connected in $H$. 

For the other direction, we show that if $G^*$ has a (not necessarily induced) subgraph $H^*$ isomorphic to $H$, then $G$ has a $t$-clique. Let $S_1=\cup_{v\in C_1}G_v$. First suppose that $H^*$ has a copy of $F$ in $S_1$. Let the vertices of this copy be $w_1,\ldots, w_z$. For each $w_i$ there is a vertex $v_i\in H$ such that $w_i\in G_{v_i}$. Now if for some $i\neq j$, $v_i=v_j$, then the induced subgraph on $\{v_1,\ldots,v_z\}$ has less than $z$ vertices, so it is $t-1$ colorable (using proper coloring). Now if we color $w_i$ the same color as $v_i$, we get a proper coloring of this copy of $F$ with $t-1$ colors, a contradiction to the chromatic number of $F$. So for each $i\neq j$, $v_i\neq v_j$. Now we show that the induced subgraph on $\{v_1,\ldots,v_z\}$ in $H$ is isomorphic to $F$. Call this subgraph $F'$. We just showed that $|V(F')|=z$. Since there is no edge between $G_{v_i}$ and $G_{v_j}$ if $v_i$ and $v_j$ are not connected, we have that $F$ is a subgraph of $F'$, and so $F'$ is not $t-1$ colorable, and since it is a subgraph of $H$, it is $t$-chromatic. If $F$ and $F'$ are not isomorphic, then $F'$ has more edges than $F$, which is a contradiction. So $F$ and $F'$ are isomorphic,
and in particular $w_i$ and $w_j$ are adjacent if and only if $v_i$ and $v_j$ are adjacent. Suppose that $w_i\in G_{v_i}$ is the copy of $w_i'$ in $G$. We show that $\{w_1',\ldots,w_z'\}$ induces a $t$-clique in $G$. Consider the coloring $f$ on $C_1$. First note that if $v_i$ and $v_j$ are adjacent vertices such that $f(v_i)=f(v_j)$, then since $w_i$ and $w_j$ are adjacent, we have $w_i'=w_j'$. Since $F_1$ has a $K_t$ minor with respect to the coloring $f$, the subgraph that each color induces is connected, and so for each $v_i$ and $v_j$ with $f(v_i)=f(v_j)=a$ we have $w_i'=w_j'$. This means that all $w_i$'s with $f(v_i)=a$ are copies of the same vertex, say $u_a$. Now take a pair of colors, $a,b\in \{1,\ldots,t\}$. There are vertices $v_i$ and $v_j$ such that $f(v_i)=a$, $f(v_j)=b$ and $v_iv_j$ is an edge in $H$. So $w_iw_j$ is an edge in $G^*$, and since $a\neq b$, $w_i'\neq w_j'$, and $w_i'w_j'$ is an edge in $G$. Since $w_i'=u_a$ and $w_j'=u_b$, we have that $u_a$ and $u_b$ are different vertices and they are adjacent in $G$. So $\{w_1',\ldots,w_z'\}=\{u_1,\ldots,u_t\}$ induces a $t$-clique in $G^*$. 

Now suppose that there is no copy of $F$ in the induced subgraph on $V(H^*)\cap S_1$ in $G^*$. For $i\in \{2,\ldots,p_F(H)\}$, let $S_i'=\cup_{v\in C_i\cap C_1} G_v$, $S_i''=\cup_{v\in C_i\setminus C_1} v^*$ and $S_i=S_i'\cup S_i''$. We prove that the collection $\{S_2\cap V(H^*),\ldots, S_{p_F(H)}\cap V(H^*)\}$ is an $F$-covering for $H^*$, which means that $p_F(H)=p_F(H^*)<r$, a contradiction.

First we show that any copy of $F$ in $H^*$ is in one of $S_i$s. Let $F^*$ with vertex set $\{w_1,\ldots, w_z\}$ be a copy of $F$ in $H^*$. For each $w_i$, there is a $v_i\in H$ where $w_i\in G_{v_i}$ and $v_i\in C_1$ if $w_i\in S_1$, or $w_i=v_i^*$ and $v_i\notin C_1$ if $w_i\notin S_1$. If $v_i=v_j$ for some $i\neq j$, then $F^*$ is $t-1$ colorable (with proper coloring): the induced graph on $\{v_1,\ldots,v_z\}$ has at most $z-1$ vertices and so it is $t-1$ colorable. Color $w_i$ the same as $v_i$. From the way we construct $G^*$ we know that if $v_i$ and $v_j$ are not connected, $w_i$ and $w_j$ are also not connected, and so this coloring of $F^*$ is proper. Since $F^*$ is $t$-chromatic, this is a contradiction. So if we call the induced graph on $\{v_1,\ldots,v_z\}\subseteq V(H)$ by $F^H$, then $|V(F^H)|=z$. We know that if $w_i$ and $w_j$ are connected, then $v_i$ and $v_j$ are connected. So $F$ is a subgraph of $F^H$, and so $F^H$ is $t$-chromatic. If $F^H$ and $F$ are not isomorphic, it means that $F^H$ has more edges than $F$, which is a contradiction. So $F^H$ and $F$ are isomorphic.
Now Since $F^*$ is not in $S_1$, wlog we can assume that $w_1\notin S_1$, and so $w_1=v_1^*$ and $v_1$ is not in $C_1$. So $F^H\notin C_1$ and there is some $i\ge 2$ such that $F^H\in C_i$. So $F^*$ is in $S_i$. 

Now we show that for each $i\ge 2$, $S_i\cap V(H^*)$ is $(K_t,F)$-minor colorable. Since $C_i$ is $(K_t,F)$-minor colorable, there is a coloring $f_i:C_i\rightarrow \{1,\ldots,t\}$ such that each induced copy of $F$ in $C_i$ has a $K_t$ minor with respect to $f_i$. Let $f_i^*:S_i\cap V(H^*)\rightarrow \{1,\ldots,t\}$ be the following coloring: For each $v\in C_i\cap C_1$, let $f_i^*(u)=f_i(v)$ for all vertices $u\in S_i\cap V(H^*) \cap G_v$. For each $v\in C_i\setminus C_1$ where $v^*\in V(H^*)$, let $f_i^*(v^*)=f_i(v)$. Now if $F^*=\{w_1,\ldots,w_z\}$ is a copy of $F$ in $S_i\cap V(H^*)$, we know that the set $F^H=\{v_1,\ldots,v_z\}$ is a copy of $F$ in $C_i$, where $w_i\in G_{v_i}$ if $w_i\in S_1$ and $w_i=v_i^*$ if $w_i\notin S_1$. Note that $f(v_i)=f_i^*(w_i)$ and $v_i$ and $v_j$ are adjacent if and only if $w_i$ and $w_j$ are adjacent. So since the subgraph induced on vertices of any color in $F^H$ is connected, the subgraph induced on any color in $F^*$ is also connected. Moreover, since in $f_i$ for any pair of colors there is an edge between one of the vertices of that color to one of the vertices of the other color, this property holds for $f_i^*$. So $S_i\cap V(H^*)$ is $(K_t,F)$-minor colorable, and so we have an $F$-covering for $H$ of size less than $p_F(H)$.
\end{proofof}

\begin{corollary}
Let $H$ be a pattern and let $t$ be the maximum chromatic number of $H$ and its complement. Then under Hadwiger conjecture, finding an induced copy of $H$ in an $n$-node graph is at least as hard as finding a $t$-clique in an $O(n)$-node graph.
\end{corollary}

\section{Induced pattern detection: Algorithms}
\label{induced-upperbound}
In this section we focus on the algorithmic part of the induced pattern detection problem, starting with some background on the problem.
First, it is a simple and folklore exercise to show that if there is a $T(n)$ time algorithm that can {\em detect} whether $G$ contains a copy of $H$, then one can also {\em find} such a copy in $O(T(n))$ time: 
 Partition the vertices $V$ of $G$ into $k+1$ equal parts (wlog $n$ is divisible by $k+1$), for every $k$-tuple of parts, use the detection algorithm in $T(nk/(k+1))$ time to check whether the union of the parts contains a copy of $H$. The moment a $k$-tuple of parts is detected to contain a copy of $H$, stop looking at other $k$-tuples and recurse on the graph induced by the union of the $k$ parts. (Stop the recursion when $n$ is constant, and brute force then.) Since every $k$ node subgraph is contained in some $k$-tuple of the parts, the algorithm is correct. The runtime is
\begin{align*}
    t(n)&\leq \sum_{i=1}^{\log_{(1+1/k)} n} (k+1)T(n (k/(k+1))^i \\ &\leq(k+1)T(n)\sum_{i=1}^{\infty} ((k/(k+1))^2)^i\leq O(T(n)).
\end{align*}
The second inequality above follows since $T(n)\geq \Omega(n^2)$ as the algorithm needs to at least read the input and the input can be dense. Because of this, for some nondecreasing function $g(n)$, $T(n)=n^2 g(n)$. Hence for any $L\geq 1$, $T(n/L)=n^2/L^2 g(n/L)\leq n^2/L^2 g(n)=T(n)/L^2$. (Without this observation about $T(n)$, the analysis would incur at most a $\log n$ factor for finding from detection.)
As finding and detection are equivalent, we will focus on the detection version of the problem.

\sloppy Recall from the introduction, $C(n,k):=M(n^{\lfloor k/3\rfloor},n^{\lceil k/3\rceil},n^{\lceil (k-1)/3\rceil})$.
Ne\v{s}etril and Poljak~\cite{Clique1} 
showed that the pattern detection problem can be reduced to rectangular matrix multiplication. In particular, when $k\equiv q \mod 3$, detecting a $k$ node pattern in an $n$ node $G$ can be reduced in $O(n^{(2k+q)/3})$ time to the product of an $n^{\lfloor k/3\rfloor} \times n^{\lceil k/3\rceil}$ matrix by an $n^{\lceil k/3\rceil} \times n^{\lceil (k-1)/3\rceil}$ matrix.

Here we first recall the approach from \cite{four-nodes}, and then generalize the ideas there to obtain an approach for all $k$ to show that (1) for all $k\leq 6$ and for all $k$-node $H$ that is not a Clique or Independent Set, $H$ can be detected in $O(C(n,k-1))$ time, whp, and (2) for all $k\ge 3$, there is a pattern that can be detected in time $O(C(n,k-1))$, whp.


\subsection{The approach from \cite{four-nodes}}
Vassilevska W. et al.~\cite{four-nodes} proposed the following approach for detecting a copy of $H$ in $G$:
\begin{enumerate}
\item First obtain a random subgraph $G'$ of $G$ by removing each vertex of $G$ independently and uniformly at random with probability $1/2$. 
\item Compute a quantity $Q$ that equals the number of induced $H$ in $G'$, modulo a particular integer $q$.
\item If $Q \neq 0\mod q$, return that $G$ contains an induced $H$, and otherwise, return that $G$ contains no induced $H$ with high probability.
\end{enumerate}

The following lemma from \cite{four-nodes} implies that (regardless of $q$), if $G$ contains a copy of $H$, after the first step, with constant probability, the number of copies of $H$ in $G'$ is not divisible by $q$.

\begin{lemma}[\cite{four-nodes}]
\label{samplinglemma}
Let $q \ge 2$ be an integer, $G,H$ be undirected graphs. Let $G'$ be a random induced subgraph of $G$ such
that each vertex is taken with probability $\frac{1}{2}$, independently. If there is at least one induced-$H$ in $G$, the number of induced-$H$ in $G'$ is not a multiple of $q$ with probability at least $2^{-|H|}$. 
\end{lemma}

Now using Lemma \ref{samplinglemma}, we can sample graph $G'$ from $G$, and with probability $2^{-k}$ we have the number of induced $H$ is not divisible by $q$. To obtain higher probability, we can simply repeat this procedure.

Hence, it suffices to provide an algorithm for counting the number of copies of $H$ modulo some integer. The approach from \cite{four-nodes} is to efficiently compute a quantity which is an integer linear combination 
$Q=\sum_{i=1}^t \alpha_i n_{H_i}$
of the number of copies $n_{H_i}$ in $G$ of several different patterns $H=H_1,H_2,\ldots, H_t$, so that some integer $q$ divides the coefficients $\alpha_i$ in front of $n_{H_i}$ for $i>1$ but $q$ does not divide $\alpha_1$. Thus, $Q=\alpha_1 n_{H}\mod q$.

Suppose that $d$ is the largest common divisor of $\alpha_1$ and $q$. Suppose that $d\neq 1$.
Since $q$ divides every $\alpha_k$ with $k>1$, $d$ must divide all $\alpha_i$. Hence,
 we could just consider $Q/d$ in place of $Q$ before taking things mod $q$. 
Thus wlog $\alpha_1$ and $q$ are coprime, and so $\alpha^{-1}$ exists in $\Z_q$. Hence, $Q\alpha^{-1}=n_{H}\mod q$, and we can use this quantity in step 2 of the approach above.

For instance, if $H$ is $K_4-e$ (the diamond), one can compute the square $A^2$ of the adjacency matrix $A$ of $G$ in $O(n^\omega)$ time, and compute $$Q=\sum_{(u,v)\in E} {A^2(u,v)\choose 2} = n_{K_4-e}+6n_{K_4},$$ so that $Q=n_{K_4-e}\mod 6$.

In prior work, the equations $Q$ were obtained carefully for each particular $4$ node pattern. 
In this section we provide a general and principled approach of obtaining such quantities that can be computed in $O(C(n,k-1))$ time for $k\leq 6$.

\subsection{Setup}
As mentioned earlier, two graphs $H$ and $H'$ are isomorphic if there is an injective mapping from the vertex set of $H$ onto the vertex set of $H'$ so that edges and non-edges are preserved. We will represent this mapping by presenting permutations of the vertices of $H$ and $H'$, i.e. for two graphs $H$ and $H'$ with vertex orders $H=(v_1,\ldots, v_t)$ and $H'=(w_1,\ldots, w_t)$, we say $H$ maps to $H'$ if for each $i$ and $j$, $(v_i,v_j)\in E(H)$ if and only if $(w_i,w_j)\in E(H')$. Note that if $H$ maps to $H'$, $H'$ maps to $H$ as well.

We refer to $k$-node graphs as patterns, and we want to detect them in $n$-node graphs. We will assume that every graph we consider is given with a vertex ordering, unless otherwise specified. We call a pattern with an ordering {\em labeled}, and otherwise, the pattern is {\em unlabeled}.
 By the subgraph $(v_1,\ldots, v_h)$ in a graph $G$, we mean the subgraph induced by these vertices, with this specified order when considering isomorphisms.

We partition all $k$-node patterns with specified vertex orders (there are $2^{{k\choose 2}}$ many of these) into classes and for each class we count the number of subgraphs in a given graph $G$ which map to one of the graphs in this class. 
Let $k' = \floor{\frac{k-1}{3}}$. For a $k$-node pattern $H=(v_0,\ldots,v_{k-1})$, define the class of $k$-node patterns $C(H)$ as follows:

\sloppy Let $F$ be the set of the following pairs of vertices: $(v_0,v_1),\ldots (v_0,v_{k'})$ (We sometimes refer to these pairs as the first $k'$ edges of $H$). Then $H'=(w_0,\ldots,w_{k-1}) \in C(H)$ if for all pairs of vertices $(v_i,v_j)\notin F$, we have $(v_i,v_j)\in E(H)$ if and only if $(w_i, w_j)\in E(H')$. In other words, all graphs in a class agree on the edge relation except possibly for the pairs in $F$.

Note that for any $H'\in C(H)$, we have $C(H')=C(H)$. So each $k$-node pattern is in exactly one class, which is obtained by changing its first $k'$ edges. Figure \ref{fig:twoclasses} shows two classes of graphs for $k=4$ (and hence $k'=1$). In this case the set $F$ consists of only one edge ($(v_0,v_1)$) and hence the graph classes are of size two. 

\begin{figure}[h]
\centering
    \begin{subfigure}[a]{0.4\textwidth}
        \centering
    \includegraphics[width=0.9\textwidth]{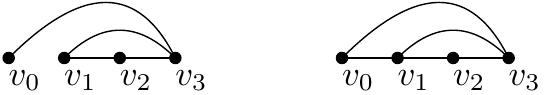}
 	\caption{class $c_1$}
        \label{c1}
    \end{subfigure}%
    \qquad
    \begin{subfigure}[a]{0.4\textwidth}
    \centering
    \includegraphics[width=0.9\textwidth]{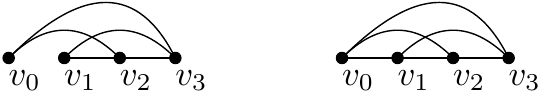}
 	\caption{class $c_2$}
    \label{c2}
    \end{subfigure}%
    \caption{Two graph classes for $k=4$. In both classes, the graphs in the class agree on all edges except the edge $v_0v_1$}
    \label{fig:twoclasses}
\end{figure}

\subsection{General Approach}
Our goal is to detect an unlabeled pattern by counting the number of patterns in different classes of graphs, which can be done as fast as the fastest algorithm for detecting $k-1$-clique (i.e. $C(n,k-1)$). Theorem \ref{generalcase} states this result formally and  we prove it at the end of this section. The graph classes possess some useful properties which we introduce in Theorem \ref{b_H} and Lemma \ref{classificationlemma} and provide their proofs in the Appendix. Using these properties, we show how to use graph classes to detect unlabeled patterns.

\begin{theorem}
\label{generalcase}
 Let $G$ be an $n$-node graph and let $c$ be one of the classes of $k$-node patterns. We can count the number of subgraphs in $G$ which map to a pattern in $c$ in $O(C(n,k-1))=O(M(n^{\floor{\frac{k-1}{3}}},n^{\ceil{\frac{k-1}{3}}},n^{\ceil{\frac{k-2}{3}}}))$ time, which is the runtime of the fastest algorithm for detecting $K_{k-1}$.
\end{theorem}

Now we need to relate unlabeled patterns to pattern classes. Each unlabeled $k$-node pattern has $k!$ possible vertex orderings. We say that an unlabeled pattern $\tilde{H}$ embeds in class $c$ if there is an ordering of vertices of $\tilde{H}$ which is in $c$. Let $U(c)$ be the set of unlabeled patterns that embed in $c$. For example, for the classes $c_1$ and $c_2$ in Figure \ref{fig:twoclasses}, $U(c_1)$ consists of the diamond (also called diam for abbreviation) and the paw (depicted in Figure \ref{unlabeled_patterns}), and $U(c_2)$ consists of the diamond and $K_4$. 
\begin{figure}
  \centering
    \includegraphics[width=0.3\textwidth]{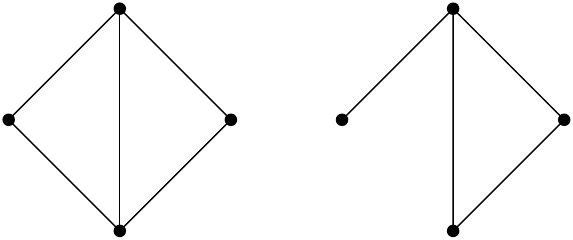}
    \caption{The diamond graph on the left and the paw graph on the right.}
	\label{unlabeled_patterns}
\end{figure}
For each unlabeled pattern $\tilde{H}$, let $\alpha_{\tilde{H}}^c$ denote the number of ways $\tilde{H}$ can be embedded in $c$, i.e. the number of vertex orderings of $\tilde{H}$ that put $\tilde{H}$ into $c$. In the example of Figure \ref{fig:twoclasses}, $\alpha_{diam}^{c_1} = 4 = \alpha_{diam}^{c_2}$,  $\alpha_{paw}^{c_1} = 2$ and $\alpha_{K_4}^{c_2} = 24$. In this example, the $\alpha^c_{\tilde{H}}$ numbers are all equal to $|Aut({\tilde{H}})|$, \footnote{$Aut(\tilde{H})$ is the automorphism group of $\tilde{H}$.} each class contains at most one labeled copy of each $H$; in general, this need not be the case.

Let $n_{\tilde{H}}$ be the number of copies of $\tilde{H}$ in $G$. We have the following corollary:

\begin{corollary}
\label{sumform}
The number of (labeled) subgraphs in $G$ which map to a pattern in $c$ is $\sum_{\tilde{H}\in U(c)} \alpha_{\tilde{H}}^c n_{\tilde{H}}$. 
\end{corollary}

The numbers $\alpha_{\tilde{H}}^c$ have some useful properties as shown in the next theorem.
\begin{theorem}
\label{b_H} 
\sloppy For any unlabeled pattern $\tilde{H}$ we have $|Aut(\tilde{H})|\, | \, \alpha_{\tilde{H}}^c$.
Moreover, for any class $c$, we have 
$$
\sum_{\tilde{H}\in U(c)} \frac{\alpha_{\tilde{H}}^c}{|Aut(\tilde{H})|}= 2^{k'}.
$$
\end{theorem}

First note that this theorem gives us upper and lower bounds on the size of $U(c)$. Each term in the above summation contributes at least $1$, so $|U(c)|\le 2^{k'}$. Moreover since $c$ has at least $k'+1$ labeled patterns which have different numbers of edges, we have $|U(c)|\ge k'+1$.  So we get the following corollary.
\begin{corollary}
\label{sizeofU}
 For any class $c$, we have $2^{k'}\ge|U(c)|\ge k'+1$.
 \end{corollary}
Define $b_{\tilde{H}}^c=\frac{\alpha_{\tilde{H}}^c}{|Aut(\tilde{H})|}$. By Theorem \ref{b_H}, the number of subgraphs in $G$ that map to a pattern in $c$ computed by Theorem \ref{generalcase} is of the following form:
\begin{equation}
\label{formofeq}
\sum_{\tilde{H}\in U(c)} b_{\tilde{H}}^c |Aut(\tilde{H})| n_{\tilde{H}}
\end{equation}
So far we showed how each pattern class relates to unlabeled patterns. Now we show how we can obtain different pattern classes from unlabeled patterns.
 
 \begin{lemma}
 \label{classificationlemma}
Let $\tilde{H}$ be an unlabeled $k$-node pattern. For an arbitrary vertex with degree at least $k'$, consider $k'$ of the edges attached to it; namely $e_1,\ldots, e_{k'}$. Let $S$ be the set of all graphs obtained by removing any number of the edges in $\{e_1,\ldots,e_{k'}\}$. Then there is a class $c$, such that $U(c)=S$. Moreover, $b_{\tilde{H}}^c=1$, and $\tilde{H}$ is the pattern with maximum number of edges in $c$.
\end{lemma}

Applying Lemma \ref{classificationlemma} to our example, consider $K_4$ as the initial pattern and consider an arbitrary edge of it. Then the set $S$ consists of the diamond and $K_4$, and so $U(c_2) = S$. Moreover, since $|Aut(K_4)|=24=\alpha_{K_4}^{c_2}$, we have $b_{K_4}^{c_2}=1$. So by Theorem \ref{b_H}, $b_{diam}^{c_2}=2-1=1$. Similarly if we consider the diamond as the initial pattern and take the edge between the degree three vertices, then the set $S$ consists of the diamond and the paw, and so $U(c_1)=S$. Moreover, since $|Aut(diam)|=4=\alpha_{diam}^{c_1}$, we have $b_{diam}^{c_1}=1$, and hence $b_{paw}^{c_1}=1.$

Now we are ready to show how to detect unlabeled patterns using graph classes. First let $B_r$ be the set of unlabeled patterns $\tilde{H}$ such that $r\mid |Aut(\tilde{H})|$. Note that we have $K_k, \bar{K_k} \in B_r$ for all $r$ such that $r| k!$ (where $K_k$ is the $k$-clique and $\bar{K_k}$ is the $k$-Independent set). For a fixed unlabeled pattern $\tilde{H}$ which is not the $k$-Independent Set or the $k$-Clique, the idea is to compute the sums of the form (\ref{formofeq}) for different pattern classes $c$, such that a linear combination of these sums gives us a sum consisting of only the terms from $\tilde{H}$ and patterns $\tilde{H'}\in B_r$ for some $r$ such that $r \not |\,|Aut(\tilde{H})|$. More specifically, we want to compute a sum of the following form:
 \begin{equation}
 \label{goaleqform}
 |Aut(\tilde{H})|n_{\tilde{H}}+\sum_{\tilde{H'}\in B_r} d_{\tilde{H'}}|Aut(\tilde{H'})|n_{\tilde{H'}}
 \end{equation}
where $d_{\tilde{H'}}$ are some integers. Then using the fact that this sum is equal to $|Aut(\tilde{H})|n_{\tilde{H}}$ modulo $r$,
by the approach of Vassilevska W. et al.~\cite{four-nodes} we can assume with constant probability that  $r\not |\, n_{\tilde{H}}$, and hence we can detect $\tilde{H}$ in $G$. 

We provide the proof of Theorem \ref{generalcase}, and in the next section we use our approach to show that for each $k$, there is a pattern that can be detected in time $O(C(n,k-1))$. Moreover, in the Appendix we show how our approach is used to prove that any $k$-node pattern except $k$-clique and $k$-independent set can be detected in $O(C(n,k-1))$ time, for $k\le 6$.

\subsection{Proof of Theorem \ref{generalcase}}
The general idea is to remove one vertex, divide the rest of the vertices into three (almost) equal parts. Then form two matrices such that the first matrix captures the subgraphs isomorphic to the removed vertex plus the first part, and the second matrix captures the subgraphs isomorphic to the removed vertex plus the second and the third part, and then use matrix multiplication to count the number of subgraphs isomorphic to the whole pattern in the host graph. We show the approach more formally below.

Let $V(G) = \{v_1,\ldots,v_n\}$. 
Let $H=(w_0,\ldots,w_{k-1})$ be an arbitrary pattern in $c$ (so $c=C(H)$). Recall that $k'=\lfloor \frac{k-1}{3}\rfloor$. Our algorithm consists of three steps. In step one, for each $t=k-k'-1$ vertices $v_{i_1},\ldots, v_{i_t}$, we count the number of vertices $u$ in $G$ such that the subgraph $(u,v_{i_1},\ldots, v_{i_t})$ in $G$ maps to the subgraph $(w_0,w_{k'+1},w_{k'+2},\ldots, w_{k-1})$ in $H$. In step two, we count the number of $k'$-tuples $(v_{j_1},\ldots, v_{j_{k'}})$ such that the subgraph $(v_{j_1},\ldots, v_{j_{k'}}, v_{i_1},\ldots, v_{i_t})$ in $G$ maps to the subgraph $(w_1,\ldots, w_{k-1})$ in $H$.  In step three, we show 
how to combine the numbers obtained in the last two steps to get the resulting value.

Before we explain each step, here is some notation. Let $k_1 = \ceil{\frac{k-1}{3}}$ and $k_2=\ceil{\frac{k-2}{3}}$. Note that $k_1,k_2\in \{k',k'+1\}$ and $k'+k_1+k_2=k-1$. Define the set $S$ to be all $t$-tuples $p=(v_{i_1},\ldots, v_{i_t})$ where the subgraph induced by $p$ maps to the subgraph $(w_{k'+1},\ldots, w_{k-1})$ in $H$. We can write each $t$-tuple $p$ with a pair of $k_1$ and $k_2$ tuples, $p'$ and $p''$; i.e. $p'=(v_{i_1},\ldots, v_{i_{k_1}})$ and $p''=(v_{i_{k_1+1}},\ldots, v_{i_{t}})$

{\bf Step one}: Construct two matrices $B$ and $C$ of sizes $n^{k_1}\times n$ and $n\times n^{k_2} $ as follows: For each $k_1$-tuple $p_1=(v_{i_1},\ldots, v_{i_{k_1}})$ and each vertex $v_h\in G$, let $B_{p_1,v_h}=1$ if the subgraph $(v_h,p_1)$ in $G$ maps to the subgraph $(w_0,w_{k'+1},\ldots, w_{k'+k_1})$ in $H$.  Otherwise set it to $0$. For each $k_2$-tuple $p_2=(v_{j_1},\ldots, v_{j_{k_2}})$ and each vertex $v_h\in G$, let $C_{v_h,p_2}=1$ if the subgraph $(v_h,p_2)$ in $G$ maps to the subgraph $(w_0,w_{k'+k_1+1},\ldots, w_{k-1})$ in $H$.  Otherwise set it to $0$. Compute $M=BC$. For any $p_1=(v_{i_1},\ldots, v_{i_{k_1}})$ and $p_2=(v_{j_1},\ldots, v_{j_{k_2}})$ such that the $t$-tuple $(p_1,p_2)\in S$, we have $M_{p_1,p_2}$ is the number of vertices $u$ such that the subgraph $(u,p_1,p_2)$ in $G$ maps to the subgraph $(w_0,w_{k'+1},\ldots,w_{k-1})$ in $H$.

\sloppy {\bf Step two:} Construct two matrices $B'$ and $C'$ of sizes $n^{k_1}\times n^{k'}$ and $n^{k'}\times n^{k_2} $ as follows: For each $k_1$-tuple $p_2=(v_{i_1},\ldots, v_{i_{k_1}})$ and each $k'$-tuple $p_1=(v_{j_1},\ldots, v_{j_{k'}})$ in $G$, let $B_{p_2,p_1}'=1$ if the subgraph $(p_1,p_2)$ in $G$ maps to the subgraph $(w_1,\ldots, w_{k'+k_1})$ in $H$.  Otherwise set it to $0$. For each $k_2$-tuple $p_3=(v_{h_1},\ldots, v_{h_{k'}})$ and each $k'$-tuple $p_1=(v_{j_1},\ldots, v_{j_{k'}})$ in $ G$, let $C_{p_1,p_3}'=1$ if the subgraph $(p_1,p_3)$ in $G$ maps to the subgraph $(w_1,\ldots, w_{k'},w_{k'+k_1+1},\ldots, w_{k-1})$ in $H$.  Otherwise set it to $0$. Compute $M'=B'C'$. For any $p_2=(v_{i_1},\ldots, v_{i_{k_1}})$ and $p_3=(v_{h_1},\ldots, v_{h_{k_2}})$ such that the $t$-tuple $(p_2,p_3)\in S$, we have $M_{p_1,p_3}'$ is the number of $k'$-tuples $p_1$ in $G$ such that the subgraph $(p_1,p_2,p_3)$ in $G$ maps to the subgraph $(w_1,\ldots,w_{k-1})$ in $H$.

{\bf Step three:} Let $r$ be the number of vertices $w_i$ in $\{w_1,\ldots,w_{k'}\}$, such that the subgraph $(w_i,w_{k'+1},\ldots,w_{k-1})$ in $H$ maps to the subgraph $(w_0,w_{k'+1},\ldots,w_{k-1})$ in $H$. 
Compute the following sum using matrices $M$ and $M'$:
\begin{equation}
\label{sumofmatrices}
\sum_{p\in S}  (M_{p',p''}-r)M_{p',p''}'
\end{equation}

\sloppy If $r=0$, by the way we constructed $M$ and $M'$, each number $M_{p',p''}M_{p',p''}'$ is the number of $k'+1$ tuples $(v_{i_0},\ldots, v_{i_{k'}})$ such that the subgraph $(v_{i_0},p',p'')$ in $G$ maps to the subgraph $(w_0,w_{k'+1},\ldots,w_{k-1})$ in $H$, and the subgraph $(v_{i_1},\ldots,v_{i_{k'}},p',p'')$ in $G$ maps to the subgraph $(w_1,\ldots,w_{k-1})$ in $H$. So the number in equation \ref{sumofmatrices} is the number of subgraphs in $G$ which map to a pattern in $c$. Now if $r>0$, then each $k'$-tuple that is counted in $M_{p',p''}'$ contains exactly $r$ vertices that are also counted in $M_{p',p''}$ and cannot be used simultaneously. So in this case, the number $(M_{p',p''}-r)M_{p',p''}'$ counts the number of $k'+1$ tuples with the property mentioned above.

\sloppy Now we analyze the running time. $M$ and $M'$ in step one and two can be computed in $O(M(n^{\floor{\frac{k-1}{3}}},n,n^{\ceil{\frac{k-2}{3}}}))$ and $O(M(n^{\floor{\frac{k-1}{3}}},n^{\ceil{\frac{k-1}{3}}},n^{\ceil{\frac{k-2}{3}}}))$ time, respectively, using rectangular matrix multiplication. By checking all $t$-tuples of vertices in $G$ in $n^{t}$ time, we can identify the set $S$, and then the sum in step three can be computed in $O(|S|)\le O(n^t)$ time. Note that $O(M(n^{\floor{\frac{k-1}{3}}},n^{\ceil{\frac{k-1}{3}}},n^{\ceil{\frac{k-2}{3}}}))\ge n^{\max{(\floor{\frac{k-1}{3}}+\ceil{\frac{k-1}{3}},\ceil{\frac{k-1}{3}}+\ceil{\frac{k-2}{3}}})}$ which is the size of the input in rectangular matrix multiplication, and also we have $t=k-1-k' \le \max{(\floor{\frac{k-1}{3}}+\ceil{\frac{k-1}{3}},\ceil{\frac{k-1}{3}}+\ceil{\frac{k-2}{3}}})$. So the total the running time is $O(M(n^{\floor{\frac{k-1}{3}}},n^{\ceil{\frac{k-1}{3}}},n^{\ceil{\frac{k-2}{3}}}))$.

\section{Patterns easier than cliques}
\label{induced-easier-patterns}
Using the approach of Section \ref{induced-upperbound}, we show that for any $k$, there is a pattern that contains a $k-1$-clique and can be detected in $O(C(n,k-1))$ time in an $n$-node graph $G$. Since this pattern has a $k-1$-clique as a subgraph, it is at least as hard as $k-1$-clique to detect, which means that the runtime obtained for it is \emph{tight}, if we assume that the best runtime for detecting $k-1$-clique is $O(C(n,k-1))$. Let $H_s^k$ be the $k$-node pattern consisting of a $(k-1)$-clique and a vertex adjacent to $s$ vertices of the $(k-1)$-clique. Assume that $s\ge \ceil{\frac{k-1}{2}}$. If $s\neq k-2$, then $|Aut(H_s^k)|=s!(k-s-1)!$. For $s=k-2$, $|Aut(H_{k-2}^k)|=(k-2)!2!$. So in all cases $|Aut(H_s^k)|$ is divisible by $s!(k-s-1)!$.

\begin{theorem}
\label{thm:easierpatterns}
Let $k$ be any positive integer, and suppose that there exists $s$, $\ceil{\frac{k-1}{2}}\le s\le k-1-\floor{\frac{k-1}{3}}$, such that $s+1$ is a prime number. Then $H_{s}^k$ can be detected in $C(n,k-1)$ time with high probability. 
\end{theorem}
\begin{proof}
Let the vertex outside the $(k-1)$-clique in $H_{s}^k$ be $v_0$. We know that if $k'=\floor{\frac{k-1}{3}}$, there are at least $k'$ vertices that are not attached to $v_0$ because $s\le k-1-k'$. Let $v_1,\ldots,v_{k'}$ be $k'$ of the vertices of the $(k-1)$-clique that $v_0$ is not attached to. 
Let $v_{k'+1},\ldots,v_k$ be the rest of the vertices. Consider the ordering $H=(v_0,v_1,\ldots, v_k)$ of $H_s^k$, and let $c=p(H)$ be the class defined by $H$. Note that $U(c)$, which is the set of unlabeled graphs that can be embedded in $c$, is $\{H_s^k,H_{s+1}^k,\ldots,H_{s+k'}^k\}$. So the equation we get from this class in time $O(C(n,k-1))$ is $Q=\sum_{i=0}^{k'} b_i |Aut(H_{s+i}^k)|n_{H_{s+i}^k}$, where $b_i$ is some integer and $b_0=1$ (by an argument similar to Lemma \ref{classificationlemma}). Since $s\ge (k-1)/2$, we have that $s+1> k-s-1$, and so $|Aut(H_s^k)|$ is not divisible by $s+1$, which means that the coefficient of $n_{H_{s}^k}$ in the equation is not divisible by $s+1$. However, for all $i\ge 1$, we have that $|Aut(H_{s+i}^k)|$ is divisible by $s+1$. So $Q$ is of the form (\ref{goaleqform}) for $r=s+1$, and hence we can detect $H_s^k$ in time $O(C(n,k-1))$ with high probability.
\end{proof}

\begin{lemma}
For any positive integer $k\ge 3$, $k\neq 14$, there exists $s$ such that $\ceil{\frac{k-1}{2}}\le s\le k-1-\floor{\frac{k-1}{3}}$ and $s+1$ is prime. 
\end{lemma}
\begin{proof}
We are going to use two theorems about prime numbers in intervals. The first one is due to Loo \cite{loo2011primes} that says for all $n>1$, there is a prime number in $(3n,4n)$. The second theorem is due to Nagura \cite{nagura1952interval} and says that for all $x\ge 25$, there is a prime number in $[x,6x/5]$. 

First suppose that $k=6t+i$ for two nonnegative integers $t$ and $i$ where $0\le i\le 5$ and $i\neq 2$. If $i<2$, let $n=t$, and otherwise let $n=t+1$. We need a prime in the interval $I=(\ceil{\frac{k-1}{2}} , k-\floor{\frac{k-1}{3}}+1)$, and since $\ceil{\frac{k-1}{2}}\le3n$ and $4n\le k-\floor{\frac{k-1}{3}}+1$, there exists such a prime by the first theorem. Now assume that $i=2$. If $t\ge 8$, then $\ceil{\frac{k-1}{2}}+1\ge 25$, and so if $x=\ceil{\frac{k-1}{2}}+1$, then $6x/5\le k-\floor{\frac{k-1}{3}}$ and so there is a prime in the interval $I$ by the second theorem. 
Now suppose that $t\le 7$ and $i=2$. For $t=1,3,4,5,6,7$, the prime numbers in the interval $I$ associated to each $k$ are $5,11,17,17,23,23$ respectively.
\end{proof}

For $k=14$, we show that we can detect $H_7^{k}$ in $O(C(n,k-1))$ time. The approach is the same as Theorem \ref{thm:easierpatterns}: we look at the class $c$ where $U(c)$ consists of $H_7^k,\ldots,H_{11}^k$ and we consider the equation $Q=\sum_{i=0}^{4} b_i |Aut(H_{7+i}^k)|n_{H_{7+i}^k}$ which can be obtained in $O(C(n,k-1))$ time, where $b_0=1$ (by an argument similar to Lemma \ref{classificationlemma}). Now note that $|Aut(H_{7+i}^k)|$ is divisible by $2^9$ for all $0<i\le 4$, and $|Aut(H_{7}^k)|$ is not divisible by $2^9$. So $Q$ is of the form (\ref{goaleqform}) for $r=s+1$, and hence we can detect $H_7^k$ in $O(C(n,k-1))$ time, and hence we have the following corollary.

\begin{corollary}
For all $k>2$, there is some $s$ where the $k$-node pattern $H_s^k$ can be detected in $O(C(n,k-1))$ time with high probability. 
\end{corollary}

\section{Detecting non-induced directed cycles}
In this Section we analyze an algorithm proposed by Yuster and Zwick~\cite{YuZw04}, obtaining the fastest algorithms for $k$-Cycle detection in sparse directed graphs, to date.

We begin by summarizing the algorithm.

\subsection{Yuster and Zwick's Algorithm} \label{alg:detectcycle}
Let $k\geq 3$ be a constant. Let $G=(V,E)$ be a given directed graph with $|V|=n,$ $|E|=m$. The algorithm will find a $k$-Cycle in $G$ if one exists. First, let us note that we can assume that $G$ is $k$-partite with partitions $V_0,\ldots,V_{k-1}$ so that the edges only go between $V_i$ and $V_{i+1\mod k}$ (for $i\in\{0,\ldots,k-1\})$. This is because we can use the Color-Coding technique~\cite{AYZ97}: if we assign each vertex $v$ a color $c(v)\in \{0,\dots,k-1\}$ independently uniformly at random and then place $v$ into $V_{c(v)}$, removing edges that are not between adjacent partitions $V_i$ and $V_{i+1\mod k}$, then any $k$-Cycle will be preserved with probability $\geq 1/k^k$. The procedure can be derandomized at the cost of a $O(\log n)$ factor in the runtime.

Now that we have a $k$-partite $m$-edge $G$, we are looking for a cycle $v_0\in V_0\rightarrow v_1\in V_1\rightarrow \ldots\rightarrow v_{k-1}\in V_{k-1}\rightarrow v_0$. Let us partition the vertices $V$ into $\log n$ degree classes: $W_j=\{v\in V~|~deg(v)\in [2^j,2^{j+1})\}$. We refer to a degree class $W_j$ by its index $j,$ for simplicity of notation.

For all $(\log n)^k$ choices of degree classes $(f_0,\ldots,f_{k-1})$ with $f_r\in \{0,\ldots,\log n\}$ for all $r$, we will be looking for a $k$-Cycle $v_0\rightarrow v_1\rightarrow\ldots \rightarrow v_{k-1}\rightarrow v_0$ such that for all $j\in \{0,\ldots,k-1\}$, $v_j\in V_j\cap W_{f_j}$ (i.e. $v_j$ has degree roughly $2^{f_j}$).

It will make sense for the degrees of the cycle vertices to be expressed in terms of the number of edges $m$. For this reason, when we are considering a $k$-tuple of degree classes  $(f_0,\ldots,f_{k-1})$, we will let $m^{d_j}=2^{f_j}$, so $d_j=f_j/\log m$, and we will be talking about degree classes $(d_0,\ldots,d_{k-1})$ instead.

Now, let us fix one of the degree classes $d=(d_0,\ldots,d_{k-1})$. There are two approaches for finding a $k$-Cycle $v_0\rightarrow v_1\rightarrow\ldots \rightarrow v_{k-1}\rightarrow v_0$ such that $v_j\in V_j$ and $v_j$ has degree roughly $m^{d_j}$:

\begin{enumerate}
\item For each $j\in\{0,\ldots,k-1\}$ we know that the number of vertices in $V_j$ of degree roughly $m^{d_j}$ is $O(m^{1-d_j})$. Thus, in $O(m^{2-d_j})$ time we can run BFS from each node in $V_j$ of such degree and determine whether there is a $k$-Cycle going through it.
\item Let $p,q\in \{0,\ldots,k-1\}$. Let's denote by $B^d_{p,q}$ the $|V_p|\times |V_q|$ Boolean matrix such that for all $v_p\in V_p, v_q\in V_q$,  $B^d_{p,q}[v_p,v_q]=1$ if and only if there is a path $v_p\rightarrow v_{p+1\mod k}\rightarrow\ldots\rightarrow v_q$ (indices mod $k$) so that each $v_r\in V_r$ and the degree of $v_r$ is roughly $m^{d_r}$.

The approach here is pick a particular pair $i,j\in \{0,\ldots,k-1\}$ and compute $B^d_{i,j}$ and $B^d_{j,i}$. Then one can find a pair of vertices $v_i\in V_i,v_j\in V_j$ such that $B^d_{i,j}[v_i,v_j]=B^d_{j,i}[v_j,v_i]=1$, if such a pair exists, at an additional cost of the number of nonzero entries in $B^d_{i,j}$ and $B^d_{j,i}$ which is dominated by the runtime of computing these matrices.
\end{enumerate}

For a fixed degree class $d=(d_0,\ldots,d_{k-1})$, let $P^d_{i,j}$ be the minimum such that $B^d_{i,j}$ can be computed in $\tilde{O}(m^{P^d_{i,j}})$ time. There are three ways to compute $B^d_{i,j}$:

\begin{itemize}
\item[(a)] Compute $B^d_{i,j-1}$ and then for every vertex $v_{j-1}\in V_{j-1}$ of degree roughly $m^{d_{j-1}}$, go through all of its outneighbors $v_j$ in $V_j$ (only of degree roughly $m^{d_j}$) and set $B^d_{i,j}[v_i,v_j]$ to $1$ for every $v_i\in V_i$ for which $B^d_{i,j-1}[v_i,v_{j-1}]=1$.
\item[(b)] Similar to above but reversing the roles of $j-1$ and $i$: Compute $B^d_{i+1,j}$ and then for every vertex $v_{i+1}\in V_{i+1}$ of degree roughly $m^{d_{i+1}}$, go through all of its inneighbors $v_{i}$ in $V_{i}$ (only of degree roughly $m^{d_i}$) and set $B^d_{i,j}[v_i,v_j]$ to $1$ for every $v_j\in V_j$ for which $B^d_{i+1,j}[v_{i+1},v_{j}]=1$.
\item[(c)] For some $r$ with $i<r<j$, compute $B^d_{i,r}$ and $B^d_{r,j}$ and compute their Boolean product to obtain $B^d_{i,j}$.
\end{itemize}

The exponent of the runtime of (a) is recursively bounded as $P^d_{i,j}\leq P^d_{i,j-1}+d_{j-1}$ as in the worst case, the number of nonzero entries in $B^d_{i,j}$ could be $\tilde{O}(m^{P^d_{i,j-1}})$. Similarly, the runtime of (b) is bounded by $P^d_{i,j}\leq P^d_{i+1,j}+d_{i+1}$. The runtime of (c) is bounded by 
$$P^d_{i,j}\leq \min_{i<r<j} \max\{P^d_{i,r},P^d_{r,j},M(1-d_i,1-d_r,1-d_j)\},$$
$M(a,b,c)$ is the smallest $g$ such that one can multiply
an $m^a\times m^b$ by an $m^b\times m^c$ matrix in $O(m^g)$ time. We will not use the known fast rectangular matrix multiplication algorithms (e.g.\cite{legallrect,legall-rect-new}) here, but for clarity instead will use the estimate $M(a,b,c)\leq a+b+c-(3-\omega)\min\{a,b,c\}$.

We get the inductive definition.

\begin{equation} \label{eq:inductive} P^d_{i,i+1}=1,~\forall j\neq i+1,~ P^d_{i,j}=\min\{P^d_{i,j-1}+d_{j-1}, P^d_{i+1,j}+d_{i+1}, \min_{i<r<j} \max\{P^d_{i,r},P^d_{r,j},M(1-d_i,1-d_r,1-d_j)\}\}.\end{equation}

For $d=(d_0,\ldots,d_{k-1})$, define
$$C_k(d_0,\ldots,d_{k-1})=\min_{0\leq i<j\leq k-1} \max\{P^d_{i,j},P^d_{j,i}\}.$$ 

The algorithm above runs in $\tilde{\Theta}(m^{c_k})$ time, where 
\[c_k = \max_{d=(d_0,\ldots,d_{k-1})} \min\left\{\min_{0\leq i\leq k-1} (2-d_i), C_k(d_0,\ldots,d_{k-1})\right\}.\]


Yuster and Zwick were only able to analyze $c_k$ for $k\leq 5$. In particular, they showed that $c_3= 2\omega/(\omega+1), c_4= (4\omega-1)(2\omega+1), c_5=3\omega/(\omega+2)$. While they were not able to analyze $c_k$ for $k>5$, using extensive numerical experiments, they came up with conjectures about the structure of $c_k$ for all odd $k$ and for $k=6$. They did not propose a conjecture for larger even $k$. 

\begin{conjecture}\label{conj:odd} \footnote{The conjecture given in \cite{YuZw04} states that $c_k = (k+1)\omega/(2\omega+k-1)$. However, we discover that $c_k \leq 2- \frac{2}{k+1} < \frac{(k+1)\omega}{2\omega+k-1}$ when $\omega > \frac{2k}{k-1}$.  }
For all odd $k\geq 3$, $c_k\leq (k+1)\omega/(2\omega+k-1);$ if $\omega \leq \frac{2k}{k-1}$, $c_k = (k+1)\omega/(2\omega+k-1).$
\end{conjecture}

\begin{conjecture}\label{conj:6}\footnote{The conjecture given in \cite{YuZw04} had a slight typo in the first case - the denominator stated there was $(4\omega+4)$ instead of $(4\omega+3)$. However, looking at the numerical experiments given to us by Uri Zwick we saw that it should be corrected, and indeed we prove that the corrected version is correct.}
\begin{equation} \label{eq:6cycruntime}
c_6 = \begin{cases} 
\frac{10 \omega -3}{4 \omega +3}, & \text{if}\ 2 \leq \omega \leq \frac{13}{6}\\
\frac{22-4\omega}{17 - 4\omega}, & \text{if}\ \frac{13}{6} \leq \omega \leq \frac{9}{4}\\
\frac{11 \omega -2}{4\omega + 5}, & \textit{if}\ \frac{9}{4} \leq \omega \leq \frac{16}{7} \\
\frac{10 - \omega}{7 - \omega}, & \textit{if}\ \frac{16}{7} \leq \omega \leq \frac{5}{2}
\end{cases}
\end{equation}
\end{conjecture}

We prove these conjectures, and in addition prove upper bounds on $c_k$ that are tight when $\omega=2$.

\subsection{The runtime of Yuster-Zwick's algorithm for finding $k$-Cycles}
Here we prove Conjectures~\ref{conj:6} and~\ref{conj:odd}, and in addition we give bounds for all even $k$ that are tight when $\omega=2$. This proves  Theorem~\ref{thm:cycles} from the introduction.
Let $C_k$ denote the $k$-Cycle.





To highlight the result for even cycles for which there wasn't even a conjectured runtime, we split it into its separate theorem:
\begin{theorem} \label{thm:evencycle}
For all even $k\geq 4$, $c_k \leq \frac{k \omega - \frac{4}{k}}{2\omega + k-2 - \frac{4}{k}}$. This bound is tight for $\omega = 2$.
\end{theorem}

\subsection{Setup: Basic Lemmas}
For simplicity's sake, we write $P_{i,j}$ for $P^{(d_0,..,d_{k-1})}_{i,j}$ when $(d_0,..,d_{k-1})$ is fixed. Here and in what follows, all indices are considered modulo $k$. Visualizing these indices as $k$ points arranging counterclockwise on a circle would make the following definitions and inequalities more intuitive.

\begin{definition}
For any index $r$, and $\delta \geq 0$, $r$ is $\delta$-low if $d_r < \delta$, and $\delta$-high otherwise.
\end{definition}

\begin{definition}
For any two indices $i,j$, let $\ell(i,j) = (j-i +1 ) \pmod k$ and $f(i,j) = \sum_{r=i}^{i + \ell(i,j) -1} d_r$. Note that $\ell(i,j) \geq 0$. When $\ell(i,j) = 0$ (i.e. $i = j+1$), $f(i,j) = \sum_{r = j+1}^j d_r = 0 $.

Repeatedly applying inequality $P^d_{w,y} \leq P^d_{w,y-1} + d_{y-1}$, which is derived from Equation (\ref{eq:inductive}), gives $P^d_{w,y} \leq 1+f(w+1,y-1)$ 
\end{definition}

\begin{lemma} \label{lemma:lowsum}
Suppose that $d_r \leq \delta \leq d_i, d_j$. Let $j_0 \in \{i,j\} $ such that $d_{j_0} = \min \{d_i, d_j\}$. 
\begin{enumerate}
    \item If $P^d_{j,i}, P^d_{i,r}, P^d_{r,j} \leq B$ and $C_k(d_0,\cdots,d_{k-1}) > B$ then $M(1-d_i,1-d_j, 1-d_r) > B$
    \item If $M(1-d_i, 1-d_j,1-d_r) \geq B \geq \omega(1-\delta)$ then $d_r + d_{j_0} \leq \omega - B - (\omega-2) \delta \leq 2 \delta$, and $d_r \leq \omega - B - (\omega-2) \delta -\delta$.
\end{enumerate}
\end{lemma}

\begin{proof}
\begin{enumerate}
    \item Suppose that $M(1-d_i,1-d_j,1-d_r) \leq B$. Using the matrix multiplication rule gives $P^d_{i,j} \leq \max\{P^d_{i,r},P^d_{r,j}, M(1-d_i,1-d_j,1-d_r)\} \leq B$.
    But then $C_k(d_0,\cdots, d_{k-1}) \leq \max \{ P^d_{i,j}, P^d_{j,i}\} \leq B$, a contradiction.
    \item Suppose that $d_r + d_{j_0} > \omega - B - (\omega-2) \delta $, then:
    $M(1-d_i,1-d_j,1-d_r) = 3 - d_i-d_j-d_r - (3-\omega)(1-\max\{d_i,d_j,d_r\})=2-(d_r+d_{j_0})+(\omega-2)(1-\max\{d_i,d_j\}) < 2 - (\omega -B - (\omega-2) \delta) + (\omega-2)(1-\delta) = B$ (contradiction).
    
    That $B \geq \omega(1-\delta) $ implies $d_r + d_{j_0} \leq \omega -B -(\omega-2) \delta \leq 2\delta$. This together with $d_{j_0} \geq \delta$ imply $d_r \leq \omega -B - (\omega-2) \delta -\delta$. 
    
\end{enumerate}
\end{proof}


\begin{lemma} \label{lemma:consecutivelow}
For any two indices $i, j$, and integer $t \leq k-2$. If $d_i, d_j \geq \delta $ and there are no $t+1$ consecutive $\delta$-low indices $r$ such that $i < r < j$ then $P_{i,j} \leq \max \{ 1 +t \delta , \omega (1-\delta) \} $. 

As a consequence, if there are no $t +1$ consecutive $\delta$-low indices then $C_k(d_0,...,d_{k-1}) \leq \max \{ 1 + t\delta, \omega (1-\delta) \} $.
\end{lemma}

\begin{proof}
Let $i=i_0,i_1,\ldots,i_z=j$ be the indices within $\{i,i+1,\ldots,j\}$ (indices mod $k$) such that $d_{i_b}\geq \delta$ for each $b\in \{0,\ldots,z\}$. Since there are no consecutive $t+1$ $\delta$-low indices, for each $b$, $P_{i_b,i_{b+1}}\leq 1+t\delta$, using the rule $P_{w,y}\leq P_{w,y-1}+d_{y-1}$. Now we can use the matrix multiplication rule to get $P_{i,j}\leq \max\{1+t\delta, M(1-\delta,1-\delta,1-\delta)\}$.
\end{proof}

\subsection{Finding odd cycles} \label{subsection:oddcycle}
Here we prove Yuster and Zwick's conjecture that the exponent $c_k$ of the runtime when $k$ is odd and $\omega \leq \frac{2k}{k-1}$ is $\omega(k+1)/(2\omega+k-1)$. When $k$ is odd and $\omega > \frac{2k}{k-1}$, using only rule 1 and 2(a), (b) in Algorithm \ref{alg:detectcycle}, one can prove $c_k \leq 2 -\frac{2}{k+1} < (k+1)\omega/(2\omega+k-1) $ (see Theorem 3.4 of \cite{AlYuZw97}).

Let $t: = \lfloor \frac{k-1}{2} \rfloor, h := k-t-1$. Note that $t \leq h \leq t+1$ and $2h \leq k$. 

Let $\delta \geq 0$ be a parameter to be specified later. Let $B: = 1 + t \delta$. Assume that $B \geq \omega (1-\delta)$. Pick arbitrary $0 \leq d_0,\cdots,d_{k-1} \leq 1$. Below, we write $C_k$ in place of $C_k(d_0,\cdots,d_{k-1})$ for simplicity. We need to prove that $C_k\leq B$.

By Lemma \ref{lemma:consecutivelow}, if there are no $t+1$ consecutive $\delta$-low indices then $C_k \leq \max\{1 + t\delta, \omega(1-\delta)\} = B$. Now, consider the case when there are at least $t+1$ consecutive $\delta$-low indices. WLOG, we can assume that there exists $s\in[0, h-1]$ such that indices $0$ and $s$ are $\delta$-high and indices $r$ are $\delta$-low for all $s+1 \leq r \leq k-1$. That there are at most $s-1 \leq h-2<t$ indices $r$ such that $0 < r < s$ and Lemma \ref{lemma:consecutivelow} implies $P_{0,s} \leq B.$ Our proof for upper bounds on $C_k$ will proceed as follow: suppose $C_k > B$, we use Lemma \ref{lemma:mainlowsum} to derive multiple inequalities of form $d_i + d_{r} \leq \omega -B - (\omega-2)\delta \leq 2\delta$ where $ r\in \{i+t,i-t\}$, then sum these inequalities together to get $f(R+1,R-1) \leq 2t\delta$, which implies $C_k \leq B$ by Lemma \ref{lemma:circlesum}.

\begin{lemma}\label{lemma:mainlowsum}
Consider indices $i,j, r$ where $0\leq i \leq j \leq s< r \leq k-1$ and $f(r+1,i-1), f(j+1,r-1) \leq t\delta$. If $C_k > B$ then $d_r \leq \omega -B - (\omega -2) \delta - \delta$, and $d_r + \min \{d_i, d_j \}\leq \omega - B - (\omega-2) \delta \leq 2\delta$.
\end{lemma}

\begin{proof}
$P_{0,s} \leq B$ by Lemma \ref{lemma:consecutivelow}. That $\max \{ f(s+1,r-1),f(r+1,k-1) \}  \leq \max\{ f(j+1,r-1), f(r+1,i-1)\} \leq t\delta$ implies $\max\{P_{s,r}, P_{r,0}\} \leq B$. Also, $d_r \leq \delta \leq d_0,d_s$, so $d_r \leq \omega -B - (\omega -2) \delta - \delta \leq 2\delta$ by Lemma \ref{lemma:lowsum}.

WLOG, assume $d_i \leq d_j$. If $d_i \leq \delta$ then $d_i + d_r\leq \omega -B -  (\omega -2) \delta - \delta + \delta = \omega - B - (\omega-2) \delta $. Else, $\delta \leq d_i \leq d_j$. Since $0\leq i \leq j \leq s$, by Lemma \ref{lemma:consecutivelow}, $P_{i,j} \leq B$. That $f(r+1,i-1), f(j+1,r-1) \leq t\delta$ implies $P_{r,i}, P_{j,r} \leq B$. Also, $d_r \leq \delta \leq d_i\leq d_j$, so $d_i + d_r \leq \omega -B -(\omega-2) \delta$ by Lemma \ref{lemma:lowsum}. 
\end{proof}

\begin{lemma} \label{lemma:circlesum}
If there exists index $R, 0\leq R \leq k-1$ such that $f(R+1,R-1) \leq 2t\delta$ then $C_k \leq B$
\end{lemma}

\begin{proof}
For every index $ r \in [R+1, R-2]$: $$P_{R,r+1} + P_{r,R} \leq (1 + f(R+1,r)) + (1+f(r+1,R-1)) = 2 + f(R+1,R-1) \leq 2+ 2 t\delta = 2B,$$ so either $P_{R,r+1} \leq B$ or $P_{r,R} \leq B$.

Note that $P_{R,R+1} =1 \leq B$, so there exists index $r^* := \max \{r | R+1 \leq r \leq R-1 \land P_{R,r} \leq B \} $. If $r^* = R-1$, then $P_{R,R-1} \leq B$ and $P_{R-1,R} = 1\leq B$ so $C_k\leq\max\{ P_{R,R-1}, P_{R-1,R}\} \leq B$. If $r^* \leq R-2$, then either $P_{R,r+1} \leq B$ or $P_{r,R} \leq B$. By definition of $r^*$, $P_{R,r+1} > B$, so $P_{r,R} \leq B$ and $C_k \leq \max \{P_{R,r}, P_{r,R} \} \leq B$.
\end{proof}

To use \ref{lemma:mainlowsum}, we need Definition \ref{def:lowarc} to ensure the preconditions, and Definition \ref{def:sequence} to get rid of the $\min \{ .,. \}$ symbol.
\begin{definition} \label{def:lowarc}
For any integer $q$, arc $(i,j)$ is $q$-low if $f(i,j) \leq \delta (\ell(i,j)-q)$ and $q$-high otherwise.
\end{definition}
\begin{lemma} \label{lemma:interval}
Consider indices $i,j$ such that $s+1 \leq i \leq j+1 \leq k$. If $(i,j)$ is $q$-low then $(i',j')$ is $q$-low for any $s+1 \leq i' \leq i$ and $j \leq j'\leq k-1$. If $(i,j)$ is $q$-high then $(i',j')$ is $q$-high for any $i \leq i'\leq j'+1 \leq j+1.$

\end{lemma}

\begin{proof}
Since $d_r \leq \delta \forall  s+1 \leq r \leq k-1$, 
$$f(i',j') = f(i',i-1) + f(i,j) + f(j+1,j') \leq \ell(i',i-1) \delta + (\ell(i,j)-q) \delta + \ell(j+1,j') \delta = (\ell(i,j)-q)\delta.$$
The second statement follows by taking the contrapositive of the first.
\end{proof}
\begin{lemma} \label{lemma:boundsum}
 For any indices $i,j$ such that $s+1 \leq i \leq j+1 \leq k$, $(i,j)$ is $0$-low. If $C_k > B$ then $(i,j)$ is $(h-s)$-high 
\end{lemma}
\begin{proof}
Our earlier assumption about $\delta$-low indices implies $d_r \leq \delta \forall s+1 \leq r \leq k-1$. Thus $$f(i,j) = \sum_{r = i}^{j} d_r \leq \sum_{r=i}^{j} \delta = \delta (\ell(i,j)-0).$$
Recall that $P_{0,s} \leq B$ by Lemma \ref{lemma:consecutivelow}. If $(i,j)$ is $(h-s)$-low then so is $(s+1,k-1)$ by Lemma \ref{lemma:interval}. But then
\begin{align*}
f(s+1,k-1) &\leq \delta(\ell(s+1,k-1) - (h-s)) 
=\delta (k-s-1 - (k-t-1 -s)) = \delta t = B-1\\
\Rightarrow C_k &\leq \max \{P_{0,s} , P^d_{s,0}\}\leq \max\{P_{0,s},1+f(s+1,k-1)\} \leq B.
\end{align*}
\end{proof}

\begin{definition} \label{def:sequence}
Define sequences $(a_n), (b_n)$ for $n \in \{0,\cdots,s\}$ as follows:

$a_0 = 0, b_0 = s, (a_n,b_n) = \begin{cases} (a_{n-1} +1, b_{n-1}) \text{ if $d_{a_n} \leq d_{b_n}$} \\ (a_{n-1}, b_{n-1}-1) \text{ else} \end{cases}$ 

Clearly, $(a_n)$ is weakly increasing, $(b_n)$ is weakly decreasing and $b_n - a_n = s-n \geq 0$. Let $T := a_s = b_s$. 
\end{definition}

\begin{theorem} \label{thm:main}
Let $p,q$ be integers in $[0,h-s-1]$. Let $\Delta: = p-q$. For every index $i$, let $i^{\Delta} := i + \Delta$. Let $m = s-h+t \leq s$. We say condition $(p,q)$ holds if $(s+1,b^{\Delta}_{m} +t-1)$ is $p$-low, $(s+1, b^{\Delta}_{m} +t)$ is $(p+1)$-high, $(a^{\Delta}_m-t+1,k-1)$ is $q$-low and $(a^{\Delta}_m-t, k-1)$ is $(q+1)$-high, and property $(p,q,B)$ holds if condition $(p,q)$ implies $C_k \leq B$.
\begin{enumerate}
\item \label{subthm:1} If condition $(p,q)$ holds and $C_k > B$ then:
\begin{enumerate}[label = (\alph*)]
\item \label{sublemma:structure} $\forall n \in \{0,..,m\}:$ \begin{align}
f(a^{\Delta}_n - t+1, a_n -1) &\leq t\delta \label{ineq:1}\\ 
f(b_n + 1, b^{\Delta}_n + t-1) &\leq t \delta \label{ineq:2}\\
d_r + d_{r^{\Delta} - t } \leq \omega -B - (\omega-2)\delta &\leq 2\delta \, \forall r, 0 \leq r < a_n \label{ineq:3}\\
d_r + d_{r^{\Delta} + t } \leq \omega -B - (\omega -2) \delta &\leq 2 \delta \, \forall r, b_n < r \leq s \label{ineq:4} 
\end{align}
\item \label{sublemma:circlesum} $f(b_m +1, a_m-1) \leq 2t \delta$
\end{enumerate}
\item \label{subthm:2} Property $(p,q,B)$ holds when $k$ is odd i.e. if $k$ is odd and condition $(p,q)$ holds, then $C_k \leq B$.
\end{enumerate}
\end{theorem}
\begin{figure}[h]
\centering
    \includegraphics[width=0.6\textwidth]{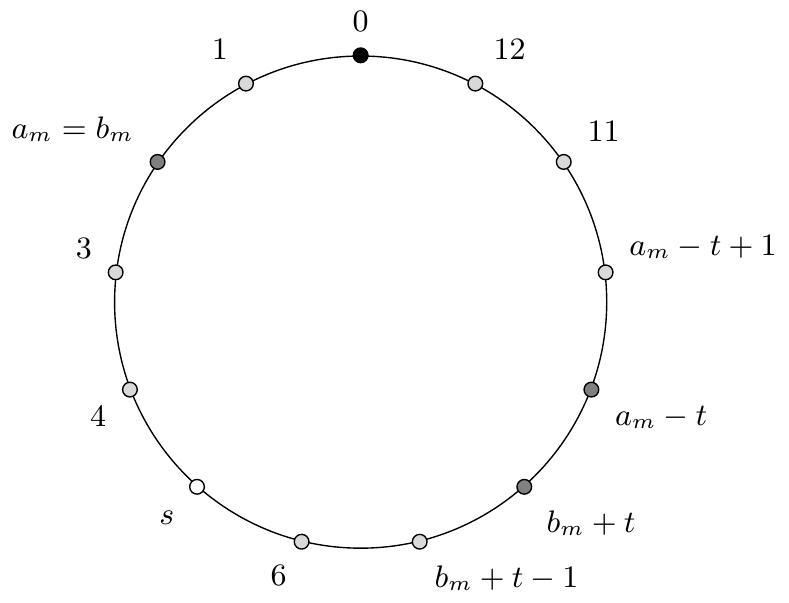}
    \caption{A visualization of the setup of Theorem \ref{thm:main} when $k=13, t=h=6, s=h-1 =5, \Delta = 0$ and $a_m=b_m=2$.}
	\label{fig:cycleDetect1}
\end{figure}
\begin{proof}
We prove \ref{thm:main}.\ref{subthm:1}\ref{sublemma:structure} by induction on $n$ for $n \in \{0,..,m\}$. First, observe a useful fact:

\begin{fact} \label{fact} $\forall n \in \{0,..,m\}$: $(s+1, b^{\Delta}_n + t-1)$ is $p$-low, $(s+1, a^{\Delta}_n-t -1)$ is $(p+1)$-high, $(a^{\Delta}_n - t +1, k-1)$ is $q$-low, $(b^{\Delta}_n+t +1, k-1)$ is $(q+1)$-high.\end{fact}

\begin{proof}[Proof of Fact \ref{fact}] 
Since $0 \leq p,q \leq h-s-1$, $$s+1 -h \leq \Delta \leq h-s-1.$$  Also, $$a^{\Delta}_n - t - 1=  a^{\Delta}_n +k-t-1 =a^{\Delta}_n + h.$$ 
Sequence $(a_n)$ defined in \ref{def:sequence} is weakly increasing and $a_m = b_m -(h-t) \leq s-(h-t)$ so: 
$$k-2\geq (s-(h-t))+(h-s-1) +h\geq a^{\Delta}_m + h \geq a^{\Delta}_n + h  \geq 0 + (s+1 -h) + h= s+1,$$
where we use $a_i^{\Delta} = a_i + \Delta.$ 

Therefore, $s+1 \leq a^{\Delta}_n - t-1< a^{\Delta}_n - t +1 \leq a^{\Delta}_m -t + 1 $ and $(a^{\Delta}_m-t+1,k-1)$ is $q$-low so $(a^{\Delta}_n - t +1, k-1)$ is $q$-low by Lemma \ref{lemma:interval}. That $b_m -a_m = s-m = h-t$ implies $s+1 \leq a^{\Delta}_n -t -1  \leq a^{\Delta}_m+h\leq b^{\Delta}_m +t$. Also, $(s+1,b^{\Delta}_m+t)$ is $(p+1)$-high so $(s+1,a^{\Delta}_n-t-1)$ is $(p+1)$-high by Lemma \ref{lemma:interval}. 

Analogously (but note that $(b_n)$ is decreasing), $(s+1, b^{\Delta}_n + t-1)$ is $p$-low and $(b^{\Delta}_n+t +1, k-1)$ is $(q+1)$-high.
\end{proof}

Now, let us proceed with the inductive proof.

Base case, $n=0$: 

By Fact \ref{fact}, $(s+1, b^{\Delta}_0+t  -1)$ is $p$-low. We have $\Delta = p-q \leq p$, so:
\begin{equation*} f(s+1,b^{\Delta}_0 + t-1) \leq \delta (\ell(s+1,b^{\Delta}_0+t-1) - p) \leq  \delta (\ell(s+1,b^{\Delta}_0+t-1) - \Delta ) < t \delta \end{equation*}

By Fact \ref{fact}, $(a^{\Delta}_0 - t  + 1, k-1) $ is $q$-low. We have $-\Delta = q-p \leq q$, so:
\begin{equation*} f(a^{\Delta}_0 -t +1, k-1) \leq \delta (\ell(a^{\Delta}_0 - t +1, a_0 -1) - q) \leq \delta (\ell(a^{\Delta}_0 - t +1, a_0 -1) - (-\Delta) ) < t \delta \end{equation*}

Since $a_0 = 0$ and $b_0 = s$, 
inequalities (\ref{ineq:1}) and (\ref{ineq:2}) are proved. Inequalities (\ref{ineq:3}) and (\ref{ineq:4}) are trivially true.

Suppose \ref{thm:main}.\ref{subthm:1}\ref{sublemma:structure} is true for $n-1$ where $m \geq n \geq 1$. WLOG, assume $d_{a_{n-1}} \leq d_{b_{n-1}}$. The case  $d_{a_{n-1}} > d_{b_{n-1}}$ is analogous.

By Definition \ref{def:sequence}, $a_n = a_{n-1}+1, b_n = b_{n-1}$. Thanks to inductive assumption, we only need to show:  $$d_{a_{n-1} } + d_{a^{\Delta}_{n-1} -t} \leq \omega - B -(\omega-2) \delta$$ and $$f(a^{\Delta}_n - t +1, a_n-1) \leq t\delta.$$ Below, write $R$ in place of $a^{\Delta}_{n-1} - t = a^{\Delta}_{n} -t-1 = a^{\Delta}_n + h$ for simplicity's sake. From proof of Fact \ref{fact}, $s+1 \leq R \leq a^{\Delta}_m + h \leq k-2$.

By the inductive assumption, $f(R+1, a_{n-1} -1) = f(a^{\Delta}_{n-1}-t+1, a_{n-1}-1)\leq t\delta$ and $f(b_{n-1}+1, b^{\Delta}_{n-1}+t-1) \leq t \delta$. By Definition \ref{def:sequence}, $h-t+1 = s-(m-1) \leq s- (n-1) = b_{n-1} - a_{n-1}$, so $R-1 = a^{\Delta}_{n-1} + h  \leq b^{\Delta}_{n-1} +t -1 $. Also, $b_{n-1} +1 \leq s+1 \leq R$, so $ f(b_{n-1}+1, R -1) \leq f(b_{n-1}+1, b^{\Delta}_{n-1}+t-1) \leq t\delta$. By Lemma \ref{lemma:mainlowsum}, $d_{a_{n-1}} + d_R \leq \omega -B - (\omega -2) \delta \leq 2 \delta$.

By the inductive assumption and the above statement, $\forall 0 \leq i \leq a_{n-1}: d_i \leq 2 \delta - d_{i^{\Delta}-t }$. Hence: 
\begin{equation} \label{eq:f0} f(0,a_{n-1}) \leq \sum_{i=0}^{a_{n-1}} (2 \delta - d_{i^{\Delta} -t }) = 2 \ell (0,a_{n-1}) \delta - \sum_{i=0^{\Delta} -t} ^{a^{\Delta}_{n-1} -t} d_i =  2 \ell (0,a_{n-1}) \delta- f(0^{\Delta}-t,R)   \end{equation}

Recall that $a_n = a_{n-1}+1$, so $R = a^{\Delta}_{n-1} - t = a^{\Delta}_{n}  - t -1$. By Fact \ref{fact}, $(s+1, R)$ is $(p+1)$-high. Since $ 0^{\Delta}-t \geq k + (s+1-h) -t \geq s+1 $, by Lemma \ref{lemma:interval}, $(0^{\Delta}-t , R)$ is also $(p+1)$-high. By Fact \ref{fact}, $ (a^{\Delta}_n - t  +1, k-1)$ is $q$-low. These together with equation (\ref{eq:f0}) imply:
\begin{align*}
&f(a^{\Delta}_n - t+1, a_n-1) = f(a^{\Delta}_n-t+1,k-1) + f(0,a_{n -1})\\
&\leq (\ell (a^{\Delta}_n - t +1, k- 1)-q) \delta + 2 \ell(0,a_{n-1} ) \delta -  f(0^{\Delta}-t , R) \\
&\leq  (\ell (a^{\Delta}_n - t  +1, k- 1)-q) \delta + 2 \ell(0,a_{n-1} ) \delta  -(\ell(0^{\Delta}-t , a^{\Delta}_{n-1} -t) - (p+1)) \delta\\
&= ((t-a_n-\Delta-1 - q)  + 2(a_{n-1} +1 ) - (a_{n-1}+1 -(p+1) )) \delta \\
& = t\delta,
\end{align*}
where the simplification in the last two lines follows from $a_n = a_{n-1}+1$ and $\Delta = p-q.$

Hence, \ref{thm:main}.\ref{subthm:1}\ref{sublemma:structure} is still true for $n$, so is true for all $n\in \{0,..,m\}.$

Now, we prove \ref{thm:main}.\ref{subthm:1}\ref{sublemma:circlesum}. Note that $b^{\Delta}_m - a^{\Delta}_m = b_m - a_m = s-m = h-t$, so $b^{\Delta}_m+t= a^{\Delta}_m +h$. By \ref{thm:main}.\ref{subthm:1}\ref{sublemma:structure}'s inequalities (\ref{ineq:3}), (\ref{ineq:4}):
\begin{equation} \label{eq:sum}
\begin{split}
2 \delta m &\geq \sum_{i=0}^{a_m-1} (d_i  + d_{i^{\Delta}-t}) + \sum_{j = b_m+1}^s (d_j + d_{j^{\Delta}+t}) = \sum_{i=0}^{a_m-1} d_i + \sum_{i= h^{\Delta} + 1}^{a^{\Delta}_m+h} d_i + \sum_{j=b_m+1}^s d_j + \sum_{j=b^{\Delta}_m+t+1}^{s^{\Delta}+t} d_j\\
&=  f(0,a_m-1) + f(h^{\Delta}+1, a^{\Delta}_m+h) + f(b_m+1,s) + f(a^{\Delta}_m+ h+1,s^{\Delta} + t) \\
&= f(0,a_m-1) + f(b_m+1,s) + f(h^{\Delta}+1, s^{\Delta}+ t)
\end{split}
\end{equation}

That $|\Delta| \leq h-s-1, 2h \leq k$ and $h-1 \leq t$ implies $$h^{\Delta} = h+\Delta \in [h-(h-s-1), h + (h-s-1)] \subseteq   [s+1, k-1],$$ and $s^{\Delta}+t+1 \in [s-(h-s-1) + t+1, s+(h-s-1)+t+1] \subseteq [s+1,k-1].$ 

By Lemma \ref{lemma:boundsum}: 
\begin{equation} \label{eq:sum1}
\begin{split} 
f(s+1, h^{\Delta}) \leq \ell(s+1, h^{\Delta} )\delta = (h^{\Delta}- s) \delta = (h+\Delta - s) \delta\\
f(s^{\Delta}+t+1, k-1) \leq \ell (s^{\Delta}+t+1, k-1) \delta = (h -s^{\Delta} )\delta = (h-\Delta -s ) \delta
\end{split}
\end{equation}

Summing equations in (\ref{eq:sum1}) with equation (\ref{eq:sum}) gives:
\begin{equation}
\begin{split}
2 t \delta &= 2\delta m + (h+\Delta -s) \delta + (h-\Delta -s) \delta \\
&\geq   f(0,a_m-1) + f(b_m+1,s) + f(h^{\Delta}+1, s^{\Delta} +t) + f(s+1, h^{\Delta})  + f(s^{\Delta} + t+ 1, k-1) \\
&= f(0,a_m-1) + f(b_m+1,s) + f(s+1,k-1) \\
&= f(b_m+1,a_m-1)\\
\end{split}
\end{equation}
 
Now, we prove \ref{thm:main}.\ref{subthm:2}. When $k$ is odd, $h=t$, $m = s$, and $a_m = b_m = a_s = b_s = T$. Suppose that $C_k > B$, then $f(T+1, T-1) \leq 2t\delta$ by \ref{thm:main}.\ref{subthm:1}\ref{sublemma:circlesum}, thus $C_k \leq B$ by lemma \ref{lemma:circlesum}.
\end{proof}
To finish the proof, we prove property $(p,q,B)$ implies $C_k \leq B$.

\begin{lemma} \label{lemma:weakenedpq} Suppose that property $(p,q,B)$ holds $\forall 0\leq p,q\leq h-s-1$.
\begin{enumerate}
\item \label{sublemma:1}
Let $p,q, \Delta$ be integers such that $0 \leq p,q \leq h-s-1,$ and $\Delta = p-q$. If $(s+1, b^{\Delta}_m + t-1)$ is $p$-low and $(a^{\Delta}_m-t, k-1)$ is $(q+1)$-high, then $C_k\leq B$. Analogously, if $(a^{\Delta}_m-t+1, k-1)$ is $q$-low and $(s+1,b^{\Delta}_m+t)$ is $(p+1)$-high then $C_k \leq B$
\item \label{sublemma:2} $C_k \leq B$
\end{enumerate}
\end{lemma}
\begin{proof}
\begin{enumerate}
\item
We prove the first statement by induction on $\Delta = p-q$. The second one follows by symmetry. We induct on $\Delta$ where $s+1 - h \leq \Delta \leq h-s-1$

Base case: $\Delta = h - s -1$. Since $0 \leq p,q \leq h-s-1$, $p = h-s-1, q =0$. By Lemma \ref{lemma:boundsum}, $(s+1, b^{\Delta}_m + t)$ is $(p+1)$-high and $(a^{\Delta}_m - t+1, k-1)$ is $q$-low, so $C_k \leq B$ because property $(p,q,B)$ holds.

Suppose Lemma \ref{lemma:weakenedpq} is true for $\Delta' = \Delta+1$, we prove it is also true for $\Delta$. $\forall$ index $i$, let $i^{\Delta'}:= i+\Delta' = i^{\Delta}+1$.

If $(s+1, b^{\Delta}_m + t) = (s+1, b^{\Delta'}_m +t-1)$ is not $(p+1)$-high, i.e. is $(p+1)$-low: Since $a^{\Delta'}_m -t > a^{\Delta}_m-t$ and $(a^{\Delta}_m-t,k-1)$ is $(q+1)$-high, by Lemma \ref{lemma:interval}, $(a^{\Delta'}_m-t, k-1)$ is $(q+1)$-high. Apply induction's assumption for $\Delta' = p+1- q $, we have $C_{k}\leq B$.

If $(a^{\Delta}_m -t+1, k-1) = (a^{\Delta'}_m -t,k-1)$ is not $q$-low, i.e. is $q$-high: Since $b^{\Delta'}_m +t-1 > b^{\Delta}_m+t-1$ and $(s+1,b^{\Delta}_m+t-1 )$ is $p$-low, by Lemma \ref{lemma:interval}, $(s+1, b^{\Delta'}_m+t-1)$ is $p$-low. Apply induction's assumption for $\Delta' = p- (q-1) $, we have $C_k \leq B$.

The remaining case is $(s+1, b^{\Delta}_m + t)$ is $(p+1)$-high and $(a^{\Delta}_m -t+1, k-1) $ is $q$-low, thus condition $(p,q)$ holds. Then $C_k\leq B$ because of property $(p,q, B)$.

Hence Lemma \ref{lemma:weakenedpq} is true.  
\item Suppose for contradiction that $C_k > B$. By Lemma \ref{lemma:boundsum}, there exists $q \in \mathbb{N}, 0 \leq q \leq h-s-1$ such that $(a_m -t, k-1)$ is $(q+1)$-high and $q$-low.

If $(s+1,b_m+t)$ is $(q+1)$-low: Since $(a_m -t, k-1)$ is $(q+1)$-high, so is $(a_m-t+1,k-1)$ by Lemma \ref{lemma:interval}. Let $p=q+1, \Delta = p-q=1$. Since $(s+1,b_m^{\Delta} +t -1) = (s+1, b_m+t)$ is $p$-low, $(a^{\Delta}_m -t,k-1) =(a_m -t+1,k-1)$ is $(q+1)$-high, $C_k \leq B$ by Sublemma \ref{lemma:weakenedpq}.\ref{sublemma:1} (contradiction).

If $(s+1, b_m+t-1)$ is $q$-high: Let $p= q-1, \Delta = p-q = -1$. Since $(s+1,b^{\Delta}_m+t)=(s+1, b_m+t-1)$ is $(p+1)$-high, $ (a^{\Delta}_m -t+1,k-1) = (a_m-t, k-1)$ is $q$-low, $C_k \leq B$ by Sublemma \ref{lemma:weakenedpq}.\ref{sublemma:1} (contradiction). So $(s+1,b_m+t-1)$ is $q$-low.

If $(a_m-t+1, k-1)$ is $q$-high: Since $(s+1,b_m+t-1)$ is $q$-low, so is $(s+1,b_m+t)$ by Lemma \ref{lemma:interval}. Let $p' = q, q'=q-1, \Delta = p'-q' = 1$. Since $(a^{\Delta}_m-t, k-1) = (a_m-t+1, k-1)$ is $(q'+1)$-high, $(s+1, b^{\Delta}_m+t-1)=(s+1,b_m+t)$ is $p'$-low, $C_k \leq B$ by Sublemma \ref{lemma:weakenedpq}.\ref{sublemma:1} (contradiction).

Hence $(s+1,b_m+t)$ is $(q+1)$-high, $(s+1,b_m+t-1)$ is $q$-low, $(a_m-t,k-1)$ is $(q+1)$-high, $(a_m-t+1,k-1)$ is $q$-low. In other words, condition $(q,q)$ holds, so $C_k \leq B$ by property $(q,q,B)$.
\end{enumerate}
\end{proof}

Suppose $k$ is odd. Set $B := 1+t\delta=\omega (1-\delta)$, then $B = (k+1)\omega/(2\omega+k-1)$. Lemma \ref{lemma:weakenedpq} and \ref{thm:main}.\ref{subthm:2} together imply $C_k \leq B$. Since the choice $d_0, \cdots, d_{k-1}$ is arbitrary, $c_k \leq B$. 

Now we show that the bound is tight and that $c_k=B$ for some choice of the degrees.
Assume $\omega \leq \frac{2k}{k-1}$. Set $d_0=\cdots= d_{k-1}=\delta$. Every matrix multiplication "costs" $M(1-\delta,1-\delta,1-\delta) = B$, and rule 1 in Algorithm \ref{alg:detectcycle} "costs" $2 - \delta \geq B$. Using only the rules 2(a), (b) "costs" at least $1+t\delta = B$ in total. Hence,   $C_k(\delta,\cdots,\delta) = B$, and the bound in Conjecture ~\ref{conj:odd} is tight for $\omega \leq \frac{2k}{k-1}$. 
\subsection{Finding 6-Cycles}
Here we prove Conjecture ~\ref{conj:6} on the runtime of Yuster-Zwick's algorithm for finding $6$-Cycles.

Let $B$ be a value dependent on $\omega$ to be specified later. Let $\delta: = \frac{B-1}{2}$. Assume that $\delta \in [0,1] $ and $B \geq \omega(1-\delta)$. Fix a degree class $(d_0,\cdots, d_5)$, and denote $C_6(d_0,\cdots, d_5)$ by $C_6$. We want to prove $C_6 \leq B$. By lemma \ref{lemma:consecutivelow}, if there exists no $3$ consecutive $\delta$-low indices then $C_6 \leq \max\{1 +2 \delta, \omega (1-\delta)\} = B$. If all indices are $\delta$-low, then $C_6 \leq \max\{P_{0,3}, P_{3,0}\} \leq 1 + \max\{d_1+d_2, d_4 +d_5 \} \leq 1 + 2 \delta = B$. Now, consider the case when there exists $3$ consecutive $\delta$-low indices, and at least $1$ $\delta$-high index. WLOG, we can assume that indices $3,4,5$ are $\delta$-low and index $0$ is $\delta$-high. Suppose for contradiction that $C_6 > B$. We will prove certain strict upper bounds on $B$, which leads to a contradiction when we set $B$ to be equal to those upper bounds. Then we will conclude that $C_6\leq B$, thus also proving the bound for $c_6$.

We need the following lemma:

\begin{lemma} \label{lemma:6cycineq}
Suppose that $d_0 \geq \delta$ and $d_3, d_4, d_5 < \delta$.
\begin{enumerate}[label=(\alph*)]
\item \label{sublemma:3} If $C_6 > B$ and $d_2 \geq \delta$ then: $B < \frac{10 \omega -3}{4 \omega +3}$ and $B < \frac{15 -2 \omega}{11 -2 \omega}$ 
\item \label{sublemma:4} If $C_6 > B$ and $d_2 < \delta$ and $d_1 \geq \delta$ then: 
\begin{itemize}
\item If $\omega \leq \frac{9}{4}$: $B < \frac{10 \omega -3}{4 \omega +3}$ and $B< \frac{22-4\omega}{17-4\omega}$
\item If $\omega > \frac{9}{4}$: $B < \frac{11 \omega -2}{4 \omega +5}$
\item If $\omega \leq \frac{5}{2}$: $B < \frac{10 - \omega}{7 -\omega}$
\end{itemize}
\item \label{sublemma:5} If $d_1, d_2 < \delta$, then $C_6 \leq B$.
\end{enumerate}
\end{lemma}

The proof of Lemma \ref{lemma:6cycineq} only involves linearly combining inequalities derived from Equation (\ref{eq:inductive}). We include the full proof of Lemma \ref{lemma:6cycineq} in the Appendix. Now, we continue on the proof of Conjecture ~\ref{conj:6}.

To show a lower bound $B$ on $c_6$, we show a tuple $(d_0, \cdots ,d_5)$, termed the "hard-case degree class", where $C_6(d_0,\cdots,d_5) = B$. Computing $C_6(d_0,..,d_5)$ given a tuple $(d_0,...,d_5)$ can be done via a constant size linear program. 

For $\delta = \frac{B-1}{2}$, $B\geq \omega(1-\delta)=\omega \frac{3-B}{2}$ iff $B \geq \frac{3 \omega}{\omega +2}$ and $\delta \in [0,1]$ iff $B \in [1,3]$. Set $B $ to be the RHS of equation (\ref{eq:6cycruntime}). It is easy to check that, for every value of $\omega \in [2,3]$, $B \geq \frac{3\omega}{\omega+2}$ and $B \in [1,3]$, so all the conditions in lemma \ref{lemma:6cycineq} are satisfied. By Lemma \ref{lemma:6cycineq}\ref{sublemma:5}, we only need to prove $C_6 \leq B$ when $d_2\geq \delta$ and when $d_2 < \delta, d_1 \geq \delta$.

\begin{enumerate}[label=(\alph*)]
\item \label{6cyc:case1}If $2\leq \omega \leq \frac{13}{6}$ then $B = \frac{10 \omega -3}{4\omega +3}$. Lemma \ref{lemma:6cycineq}\ref{sublemma:3}, \ref{sublemma:4} imply $B < \frac{10 \omega -3}{4 \omega +3}$, which is a contradiction. So $C_6 \leq B$ as needed. The "hard-case degree class" is $(\frac{4\delta}{3}, \delta, \delta, \frac{2\delta}{3}, \frac{2\delta}{3}, \frac{2\delta}{3})$.
\item If $\frac{13}{6} \leq \omega \leq \frac{9}{4}$ then $B = \frac{22 -4 \omega}{17 -4\omega}$. Lemma \ref{lemma:6cycineq}\ref{sublemma:3}, \ref{sublemma:4} imply $B < \max\{ \frac{15-2\omega}{11-2\omega},\frac{22-4\omega}{17-4\omega}\} = \frac{22 -4\omega}{17-4\omega}$, which is a contradiction. The "hard-case degree class" is $(2-B, \frac{7B-10}{4}, \frac{6-3B}{4}, \frac{2-B}{2}, \frac{2-B}{2}, 2B-3 )$.
\item If $\frac{9}{4} < \omega \leq \frac{16}{7}$ then $B= \frac{11 \omega -2}{4\omega + 5}$. Lemma \ref{lemma:6cycineq}\ref{sublemma:3}, \ref{sublemma:4} imply $B < \max\{ \frac{10\omega-3}{4\omega+3},\frac{11\omega-2}{4\omega + 5}\} = \frac{11\omega -2 }{4\omega +5}$, which is a contradiction. The "hard-case degree class" is $(\frac{8 \delta}{7}, \frac{8\delta}{7}, \frac{6\delta}{7}, \frac{4\delta }{7}, \frac{4\delta}{7}, \frac{6\delta}{7})$.
\item If $\frac{16}{7} \leq \omega \leq \frac{5}{2}$ then $B = \frac{10- \omega }{7 -\omega } $. Lemma \ref{lemma:6cycineq}\ref{sublemma:3}, \ref{sublemma:4} imply $B < \max\{\frac{15-2\omega}{11-2\omega},\frac{10-\omega}{7-\omega }\} \leq \frac{10 -\omega}{7-\omega}$, which is a contradiction. The "hard-case degree class" is $(2-B, 2-B,  2B-3,\frac{2-B}{2} ,\frac{2-B}{2} , 2B-3)$.
\end{enumerate}

\subsection{Finding even cycles}
Here we analyze the algorithm when $k$ is even. We will show that our analysis is tight when $\omega=2$.

When $k$ is even, $h = t+1$. Let $\beta := \delta\frac{t}{h} < \delta$, then $B= 1+t\delta = 1+h\beta$. Set $B = (\omega-2) (1-\delta) + (1-\delta) + (1-\beta) \geq \omega(1-\delta)$. Note that, $B=\frac{k \omega - \frac{4}{k}}{2\omega + k-2 - \frac{4}{k}}$. To prove Theorem \ref{thm:evencycle}, we fix arbitrary $d_0,\cdots, d_{k-1}$ and prove $C_k:=C_k(d_0,\cdots, d_{k-1}) \leq B$. We define $s$, $(a_n)$, $(b_n)$ same as in Subsection \ref{subsection:oddcycle}. If $s = 0$ then $P_{0,h} \leq 1 + \ell(1,h-1) \delta = B$ and $P_{h,0} \leq 1 + \ell(h+1,k-1) \delta = B$, so $C_k\leq B$.  Below, we assume that $s \geq 1$. Note that $\omega -B - (\omega-2)\delta = \delta +\beta$. As in Subsection \ref{subsection:oddcycle}, we will sum inequalities of the form $d_i + d_r \leq \delta+\beta$ to get $f(b_s+1,a_s-1) \leq t \delta + h \beta = 2 t\delta$, then use Lemma \ref{lemma:circlesum} and \ref{lemma:weakenedpq} to conclude that $C_k \leq B$.  
\begin{lemma} \label{lemma:evenlow}
Suppose that $C_k > B$.
\begin{enumerate}
\item \label{sublemma:lowsingle} $\forall r, h \leq r \leq s+h$: $d_r \leq \beta$
\item \label{sublemma:lowpair} For every index $r$, $s+1 \leq r \leq h-1: d_r + d_{r+h} \leq \delta + \beta$. Equilvalently, for every index $r$, $s+h+1 \leq r \leq k-1: d_r + d_{r-h} \leq \delta + \beta$
\item \label{sublemma:lowinterval} If index $i$ satisfied $s+1 \leq i \leq i+s-1 \leq k-1$, then $f(s+1,i-1) + f(i+s,k-1) \leq (t-s) \delta + (h-s) \beta$ 
\end{enumerate}
\end{lemma}
\begin{proof}
\begin{enumerate}
\item Since $s+1\leq h r \leq h+s \leq h+h-1 = k-1$, $f(s+1, r-1) \leq \ell(s+1,r-1) \delta\leq \ell(s+1,h+s-1)\delta =t\delta$ and $f(r+1,k-1) \leq \ell(r+1,k-1) \delta \leq \ell(h+1,k-1)=t \delta$. By lemma \ref{lemma:mainlowsum}, $d_r \leq \beta$.
\item If $d_r \leq \beta$ then $d_r + d_{r+h} \leq \beta + \delta$. 

If $d_r \geq \beta$ then $M(1-d_r,1-d_0,1-d_s) \leq B$. Clearly, $P_{s,r} \leq 1+ f(s+1,r-1)\leq 1+ t\delta = B$, and $P_{0,s} \leq B$ by lemma \ref{lemma:consecutivelow}, so $P_{0,r} \leq \max\{P_{0,s}, P_{s,r}, M(1-d_r,1-d_0,1-d_s) \} \leq B$. Clearly, $P_{r,r+h}, P_{r+h,0} \leq B$. So, by Lemma \ref{lemma:lowsum}, $M(1-d_r, 1-d_{r+h}, 1-d_0) \geq B$. Since $d_r, d_{r+h} \leq \delta \leq d_0$, $d_r + d_{r+h} \leq \omega -B - (\omega-2)d_0 \leq \omega-B -(\omega-2)\delta= \delta +\beta$.
\item We have two cases: either $i \leq h$ or $i+s-1 \geq h+s $. These two cases are symmetrically equivalent. We can assume $i \leq h$; the remaining case is analogous.
Since $i-1 \leq h-1$, by \ref{lemma:evenlow}.\ref{sublemma:lowpair}:
\begin{equation*} f(s+1, i-1) + f(s+h+1, i+h-1) = \sum_{r = s+1}^{i-1} (d_r + d_{r+h}) \leq (i-s-1) (\delta+\beta)\end{equation*}
Since $i+s \leq h+s$, by \ref{lemma:evenlow}.\ref{sublemma:lowsingle}, $f(i+s,h+s) \leq \ell(i+s,h+s) \beta =  (h-i+1) \beta$.

Since $s+1\leq i+h \leq 2h \leq k$ and $d_r \leq \delta \forall s+1 \leq r \leq k-1$, $f(i+h, k-1) = \sum_{r=i+h}^{k-1} d_r \leq \ell(i+h,k-1)\delta = (k-i-h) \delta$. Hence:
\begin{align*} 
&f(s+1,i-1) + f(i+s,k-1)\\
&= f(s+1,i-1) + f(i+s,h+s) +f(h+s+1,i+h-1) + f(i+h,k-1) \\
&\leq (i-s-1) (\delta+\beta) + (h-i+1) \beta + (k-i-h) \delta = (t-s) \delta + (h-s) \beta \end{align*}

\end{enumerate}
\end{proof}
\begin{lemma} \label{lemma:evenmain}
Consider integers $p,q \in [0, h-s-1]$ and $ \Delta: = p-q$. Suppose condition $(p,q)$ holds. 

Suppose $d_r + d_{r^{\Delta}+t} \leq \delta + \beta \forall b_s < r \leq b_m$. If 
$\exists n \in \{0,\cdots, m\}: f(a^{\Delta}_n - h +1, a_n-1) \leq t \delta$ and $\forall r\in [ a_n , a_s): d_r + d_{r^{\Delta}-h} \leq \delta+ \beta,$ then $C_k \leq B.$

Suppose  $d_r + d_{r^{\Delta}-t} \leq \delta+\beta \forall a_m\leq r < a_s$. If $\exists n \in \{0,\cdots, m\}: f(b_n +1, b^{\Delta}_n + h -1) \leq t \delta$ and $\forall r\in (b_s, b_n]: d_r + d_{r^{\Delta}+h} \leq \delta+ \beta,$ then $C_k \leq B$.
\end{lemma}
\begin{proof}
We prove the first statement by induction on $n \in \{0,\cdots, m\}$. The second statement is analogous. Suppose for contradiction $C_k > B$. By \ref{thm:main}.\ref{subthm:1}\ref{sublemma:structure}, $d_r + d_{r^{\Delta}+t} \leq \delta +\beta \forall b_m < r \leq s$, so  $d_r + d_{r^{\Delta}+t} \leq \delta + \beta \forall b_s < r \leq s$, so 
\begin{equation} \label{ineq:boundb} f(b_s+1,s) + f(b^{\Delta}_s+1+t,s^{\Delta}+t) \leq (s-b_s) \end{equation}

Base case: $n = 0$. Since $\forall r, a_0 \leq r< a_s: d_r + d_{r^{\Delta}-h} \leq \delta+ \beta$, $f(a_0, a_s-1) + f(a^{\Delta}_0 - h, a^{\Delta}_s - h -1) \leq (a_s -a_0) (\delta + \beta)$. Since $a^{\Delta}_s - h-1 = b^{\Delta}_s+k-h-1 = b^{\Delta}_s +t$, $ f(a^{\Delta}_0 - h, a^{\Delta}_s - h -1)+  f(b^{\Delta}_s+1+t,s^{\Delta}+t) = f(a^{\Delta}_0 -h, s^{\Delta}+t)$. Combining these with (\ref{ineq:boundb}) gives $f(0, a_s-1) + f(a^{\Delta}_0-h, s^{\Delta}+t) + f(b_s+1, s) \leq (s-b_s+a_s -a_0) (\delta +\beta) = s (\alpha+\beta) $.

Since $a_0=0, |\Delta |\leq h-s-1$, $s+1\leq a^{\Delta}_0-h \leq 0^{\Delta} -h +s-1 = s^{\Delta} + t \leq k-1$. By \ref{lemma:evenlow}.\ref{sublemma:lowinterval}, $f(s+1,a^{\Delta}_0 -h-1) + f( s^{\Delta}+t+1, k-1)\leq (t-s)\delta + (h-s)\beta$. So, $f(b_s+1,a_s-1) \leq s (\delta +\beta) + (t-s)\delta + (h-s)\beta =t \delta + h \beta = 2t \delta$. Since $a_s = b_s = T$, $C_k \leq B$ by lemma \ref{lemma:circlesum}.

Suppose Lemma \ref{lemma:evenmain} is true for $n-1\leq m-1$. We show that it is still true for $n$.

If $a_n = a_{n-1}$ then we are done. Else, $a_n = a_{n-1} +1$, $b_n = b_{n-1}$ and $d_{a_{n-1}} \leq d_{b_{n-1}}$. We show $f(a^{\Delta}_{n-1}-h+1, a_{n-1}-1) \leq t\delta$ and $d_{a_{n-1}} + d_{a^{\Delta}_{n-1} -h} \leq \delta + \beta$, then conclude that $C_k \leq B$ using the inductive assumption..

Consider $R = a^{\Delta}_n - h = a^{\Delta}_{n-1} -h +1$. Note that $b_n -a_n = s-n \geq s-m =1$, so $R-1=a^{\Delta}_n -h -1 = a^{\Delta}_n +t \leq b^{\Delta}_n +t-1$, so $f(b_n+1 ,R-1) \leq f(b_n+1,b^{\Delta}_n +t-1) \leq t \delta $ by \ref{thm:main}.\ref{subthm:1}\ref{sublemma:structure}. Clearly, $f(R+1, a_n-1) = f(a^{\Delta}_n-h+1, a_n-1) \leq t \delta$. So $d_R \leq \beta$ by Lemma \ref{lemma:mainlowsum}.

Note that $f(a^{\Delta}_{n-1}-h+1, a_{n-1}-1)  = f(a^{\Delta}_n-h+1, a_n-1) + d_R - d_{a_{n-1}} \leq t \delta + (d_R -d_{a_{n-1}})$.

If $d_R \leq d_{a_{n-1}}$ then $f(a^{\Delta}_{n-1}-h+1, a_{n-1}-1)  \leq t\delta$. Since $b_n = b_{n-1}$ and $a^{\Delta}_{n-1} - h -1 < a^{\Delta}_n -h-1 = R-1$, $f(b_{n-1}+1, a^{\Delta}_{n-1} - h -1) \leq f(b_n+1, R-1) \leq t\delta$. So $d_{a_{n-1}} + d_{a^{\Delta}_{n-1}-h} \leq \delta +\beta $ by Lemma \ref{lemma:mainlowsum}, thus $C_k \leq B$ by inductive assumption.

If $d_{a_{n-1}} \leq d_R \leq \beta$: Since $d_r + d_{r-h} \leq \delta +\beta \forall a_n \leq r < a_s$. By \ref{thm:main}.\ref{subthm:1}\ref{sublemma:structure}, $d_r + d_{r-t} \leq \delta +\beta \forall 0 \leq r < a_{n-1} \leq a_m$. Hence, $f(0, a_{n-1} -1) +  f(0^{\Delta}-t, a^{\Delta}_{n-1}-1-t) \leq a_{n-1}(\delta +\beta)$ and $f(a_n, a_s-1) + f(a^{\Delta}_n-h, a^{\Delta}_s-1-h) \leq (a_s-a_n) (\delta +\beta)$. Since $0^{\Delta}-h \in [s+1,k-1]$, $d_{a_{n-1}} + d_{0^{\Delta}-h} \leq \beta +\delta$. Summing these three inequalities with (\ref{ineq:boundb}) gives: $f(0,a_s-1) + f(0^{\Delta}-h, s^{\Delta}+t)+ f(b_s+1, s) \leq s (\delta+\beta)$. So $C_k \leq B$ by same argument as in base case.
\end{proof}

Now, we prove that property $(p,q,B)$ holds, then conclude that $C_k \leq B$ using lemma \ref{lemma:weakenedpq}. 
\begin{proof} 
Suppose for contradiction that condition $(p,q)$ holds but $C_k > B$. By \ref{thm:main}.\ref{subthm:1}\ref{sublemma:circlesum}, $f(b_m+1, a_m-1) \leq 2t \delta$. Note that $m=s-1$ and $b_m -a_m =1$. Let $R:= a^{\Delta}_m-h$, then $R+1 = b^{\Delta}_m+h$. Since $ f(R+1,a_m-1) + f(b_m+1, R) =f(b_m+1,a_m-1)\leq 2t\delta$, either $f(a^{\Delta}_m -h+1,a_m-1) \leq t\delta$ or $f(b_m+1,b^{\Delta}_m+h-1) \leq t\delta$. WLOG, assume $d_{a_m} \leq d_{b_m}$, then $a_s = a_m+1, b_s = b_m$.

If $f(a^{\Delta}_m -h+1,a_m-1) \leq t\delta$: Note that $f(b_m+1, a^{\Delta}_m-h-1) = f(b_m+1, a^{\Delta}_m+t) = f(b_m+1, b^{\Delta}_m+t-1) \leq t\delta$ by \ref{thm:main}.\ref{subthm:1}\ref{sublemma:structure}. So $d_{a_m} + d_{a^{\Delta}_m -h} \leq \delta +\beta$ by lemma \ref{lemma:mainlowsum}. First statement of lemma \ref{lemma:evenmain} (note that $b_s=b_m$) for $n = m$ gives $C_k \leq B$ (contradiction).

If $f(b_m+1,b^{\Delta}_m+h-1) \leq t\delta$: Note that $f(b^{\Delta}_m+h+1,a_m -1) = f(a^{\Delta}_m-t+1, a_m-1) \leq t \delta$ by \ref{thm:main}.\ref{subthm:1}\ref{sublemma:structure}. So $d_{a_m} + d_{a^{\Delta}_m-t} = d_{a_m} + d_{b^{\Delta}_m +h} \leq \delta + \beta$. Second statement of lemma \ref{lemma:evenmain} for $n=m$ gives $C_k \leq B$ (contradiction).
\end{proof}
Assume $\omega =2$. Set $d_0 = 2 \beta, d_1 =\cdots = d_t = \delta, d_{t+1}=\cdots=d_{k-1} = \beta$. Note that $M(1-2\beta, 1-\delta, 1-\beta) = M(1-\delta,1-\delta,1-\beta) = B$, $2-\delta \geq 2-2\beta \geq B$ and $t\delta = 2\beta + (t-1) \beta = B-1$. Thus, using rule 1. or 2.(c) "costs" at least $B$, and using only rules 2.(a) and 2.(b) costs at least $B$ in total. So $C_k(d_0,\cdots,d_{k-1}) =B$, thus the bound is tight for $\omega =2$.

\section{Acknowledgments} We would like to thank Uri Zwick for pointing us to the open problem of analyzing the Yuster-Zwick algorithm.
\newpage
\bibliographystyle{plain}
\bibliography{refs}
\section{Appendix}
\subsection{Induced pattern detection for $k\le 6$}

When $k\in \{5,6\}$, we have $k'=1$. Consider a class $c$. By Corollary \ref{sizeofU}, $U(c)$ has exactly two patterns which differ in only one edge $e$, namely $\tilde{H}$ and $\tilde{H}\setminus{e}$. By Theorem \ref{b_H}, $b_{\tilde{H}}^c+b_{\tilde{H}\setminus(e)}^c=2$, so $b_{\tilde{H}}^c=b_{\tilde{H}\setminus(e)}^c=1$. So for any class $c$ and any unlabeled pattern $\tilde{H}$ that embeds in $c$, we have $\alpha_{\tilde{H}}^c=|Aut(\tilde{H})|$. Moreover by Lemma \ref{classificationlemma}, there is some class $c$ such that $U(c)$ consists of $\tilde{H}$ and $\tilde{H}\setminus\{e\}$, where $e$ is an arbitrary edge in $\tilde{H}$.

Hence by Theorem \ref{generalcase} and Corollary \ref{sumform}, for any unlabeled pattern $\tilde{H}$ which has at least one edge, we can compute $n_{\tilde{H}}|Aut(\tilde{H})|+n_{\tilde{H}\setminus e}|Aut(\tilde{H}\setminus e)|$ in $O(M(n,n^2,n))$ time for $k=5$ and $O(M(n,n^2,n^2))$ time for $k=6$. Now we give an algorithm which detects any fixed pattern $\tilde{H}$ in a graph $G$, where $\tilde{H}$ is not the $k$-Clique or the $k$-Independent Set.

Let $e_1,\ldots,e_h$ be an arbitrary permutation of all the edges of $\tilde{H}$. Let $\tilde{H}_i=\tilde{H}\setminus\{e_1,\ldots,e_{i-1}\}$ where $\tilde{H}_1=\tilde{H}$. Compute $q_i= n_{\tilde{H}_i}|Aut(\tilde{H}_i)|+n_{\tilde{H}_{i+1}}|Aut(\tilde{H}_{i+1})|$. Compute $Q= \sum_{i=1}^h (-1)^i q_i$. In fact, $Q=n_{\tilde{H}}|Aut(\tilde{H})|+(-1)^h n_{\tilde{H}_{h+1}}|Aut(\tilde{H}_{h+1})|$, which is of the form (\ref{goaleqform}) for $r=k!$, since $\tilde{H}_{h+1}$ is the $k$-Independent Set. So we can detect all $5$-node patterns in time $O(M(n,n^2,n))\in O(n^{\omega+1})$, and all $6$-node patterns in time $O(M(n,n^2,n^2))\in O(n^{\omega+2})$.



 
\subsection{Omitted proofs}
\paragraph{Proof of Theorem \ref{b_H}.}
Let $H=(w_0,\ldots,w_{k-1})$ be an arbitrary pattern in $c$. Define $b_{\tilde{H}}^c$ to be the number of ways we can specify the edges $w_0w_1,\ldots, w_0w_{k'}$ so that the resulting vertex order maps to a vertex order of $\tilde{H}$. Note that this is independent of the choice of $H$, because all edges except the $k'$ edges mentioned are the same for all $H\in c$. For each of these $b_{\tilde{H}}^c$ vertex orderings, we can apply $|Aut(\tilde{H})|$ automorphisms to get a different ordering that maps to it. So all these orderings make the $\alpha_{\tilde{H}}^c$ possible ways $\tilde{H}$ can be embedded in $c$; hence $\alpha_{\tilde{H}}^c=b_{\tilde{H}}^c\cdot |Aut(\tilde{H})|$. Now note that the total number of ways we can specify the $k'$ edges $w_0w_1,\ldots,w_0w_{k'}$ is $2^{k'}$, so $\sum_{\tilde{H}\in U(c)} \frac{\alpha_{\tilde{H}}^c}{|Aut(\tilde{H})|}=\sum_{\tilde{H}\in U(c)} b_{\tilde{H}}^c =  2^{k'}$

\paragraph{Proof of Lemma \ref{classificationlemma}.}
Let $H=(w_0,w_1,\ldots,w_{k-1})$ be an ordering of the vertices of $\tilde{H}$ such that $e_i=w_0w_i$ for each $i\in \{1,\ldots,k'\}$. Now each pattern $H'\in C(H)$ differs from $H$ only in those $k'$ edges, so the unlabeled version of $H'$ is obtained from $\tilde{H}$ by removing some of $e_i$ edges. So $C(H)\subseteq S$. Now consider $\tilde{H'} \in S$. Since $\tilde{H'}$ is obtained from $\tilde{H}$, we can consider the same ordering of vertices for it. Call this vertex order $H'$. So $H'$ and $H$ differ only in the $k'$ first edges, so $H'\in C(H)$. Hence $S\subseteq U(c)$ which shows that $U(c)=S$.  

Now since the number of ways we can embed $H$ in class $c$ is $1$ (we have to put an edge between all the $k'$ pairs of vertices), we have $b_{\tilde{H}}^c=1$.

\paragraph{Proof of Lemma \ref{lemma:6cycineq}.}
For $ r \in \{0,\cdots,5\}$: If $B<C_6$, we get $B < C_6 \leq 2 - d_r$, $d_r < 2 - B$. Note that $2 \leq \omega \leq 3$.
\begin{enumerate}[label=(\alph*)]
\item 
By Lemma \ref{lemma:consecutivelow}, $P_{0,2} \leq \max\{1+2\delta, \omega(1-\delta)\} \leq B$. This and $B < C_6$ imply $B < P_{2,0}$. Hence $1+2\delta = B < P_{2,0} \leq 1 + d_3 + d_4 + d_5$, so $d_3 +d_4 + d_5 \geq 2\delta = -(B-1)$.

For $r \in \{3,4,5\}$, $\max\{P_{r,0}, P_{2,r}\} \leq 1 + \max\{f(r+1,5),f(3,r-1)\} \leq 1+2 \delta=B$. By lemma \ref{lemma:lowsum}, $M(1-d_0,1-d_2,1-d_r) > B $, which means: 
\begin{equation} \label{ineq:6cyc3low1} (\omega -2) d_0 + d_2 +d_r < \omega-B \end{equation} 
Summing (\ref{ineq:6cyc3low1}) for $r \in \{3,4,5\}$ with $-(d_3 + d_4 + d_5) < -(B-1)$ gives: 
\begin{equation} (\omega -2) d_0 + d_2 < \omega - B - \frac{B-1 }{3}   \end{equation}
Summing (\ref{ineq:6cyc3low1}) for $r = 3$ with $-(\omega-2)d_0 \leq -(\omega -2) \delta$ gives: 
\begin{align*} &d_2 + d_3 \leq \omega  - (\omega-2) \delta - B \leq \omega - (\omega-2)\delta - \omega(1-\delta) = 2 \delta\\
&\Rightarrow P_{1,4} \leq 1 + d_2 + d_3 \leq 1 + 2 \delta = B \Rightarrow B < C_6 \leq \max\{P_{1,4}, P_{4,1}\} \leq \max\{B, P_{4,1}\} \\
&\Rightarrow B < P_{4,1} \leq 1 + d_5 + d_0 \Rightarrow 2 \delta < d_0 + d_5
 \end{align*}  
Summing (\ref{ineq:6cyc3low1}) for $r= 5$ with $-(d_5 + d_0) < -2 \delta = -(B-1)$ gives: \begin{equation}\label{ineq:6cyc3low2} (\omega -3) d_0 + d_2 < \omega - B - (B-1)\end{equation}

Summing (\ref{ineq:6cyc3low1}) for $r= 5$, $-(d_5 + d_0) < -(B-1)$ and $d_0 < 2-B$ gives: 
\begin{equation} \label{ineq:6cyc3low3} (\omega -2) d_0 + d_2 < \omega - B - (2B -3) \end{equation}

Multiplying (\ref{ineq:6cyc3low1}) with $3 - \omega \geq 0$, (\ref{ineq:6cyc3low2}) with $\omega-2 \geq 0$ and summing them gives:
$d_2 \leq \frac{5 \omega}{3} - B \frac{2 \omega}{3} -1$

Since $d_2 > \delta = \frac{B-1}{2}$, $\frac{B-1}{2} < d_2 \leq  \frac{5 \omega}{3} - B \frac{2 \omega}{3} -1$, which implies $B < \frac{10 \omega -3}{4 \omega +3}$

Multiplying (\ref{ineq:6cyc3low3}) with $3 - \omega \geq 0$, (\ref{ineq:6cyc3low2}) with $\omega-2 \geq 0$ and summing them gives:
$d_2 \leq 7 - \omega - B (5-\omega)$. Thus $\frac{B-1}{2} < d_2 \leq 7 - \omega - B (5-\omega)$, which implies $B < \frac{15 - 2 \omega}{11 - 2\omega}$
\item 
$P_{3,0} \leq 1+d_4 + d_5 \leq 1+2\delta = B \Rightarrow B < P_{0,3}\leq 1 + d_1 + d_2 \Rightarrow 2 \delta < d_1 + d_2$. Analogously, $2 \delta < d_0 + d_5$

For $r \in \{0,1\}$: If $d_r \geq 2 \delta$ then $C_6 \leq 2- d_r \leq 2 (1-\delta) \leq \omega(1-\delta) \leq B$. So $d_r \leq 2 \delta$, thus $\max \{P_{5,1}, P_{0,2} \} \leq \max\{1+d_0,1+d_1\} \leq B$. Clearly $P_{1,2} = P_{5,0}= 1 \leq B, P_{2,5} \leq 1 + d_3 + d_4 \leq B$. So $P_{5,r},P_{r,2},P_{2,5} \leq B$. By Lemma \ref{lemma:lowsum}, $ M(1-d_r,1-d_5, 1-d_2) > B $, which means:
\begin{equation} \label{ineq:6cyc4low1} (\omega-2)d_r + d_2 + d_5 < \omega - B \end{equation}
Summing (\ref{ineq:6cyc4low1}) for $r=0$ with $-(\omega-2)(d_0 + d_5) \leq -(\omega-2)2 \delta = -(\omega-2) (B-1)$ gives: 
\begin{equation} \label{ineq:6cyc4low1a} d_2 + (3-\omega) d_5 < \omega - B - (\omega-2)(B-1)  \end{equation}

Summing (\ref{ineq:6cyc4low1}) for $r=1$ with $-(\omega-2)(d_1 + d_2) \leq -(\omega-2)2 \delta = -(\omega-2) (B-1)$ gives: \begin{equation} \label{ineq:6cyc4low1b} d_5 + (3-\omega) d_2 < \omega - B - (\omega-2) (B-1) \end{equation}
For $r \in \{3,4\}$: $\max\{P_{1,r}, P_{r,0} \} \leq1+ \max\{f(2,r-1), f(r+1,5)\} \leq 1 + 2\delta = B$, and $P_{0,1} = 1 \leq B$. By Lemma \ref{lemma:lowsum}, $M(1-d_0,1-d_1,1-d_r) > B$, which means:
\begin{equation} \label{ineq:6cyc4low3} (\omega -2)d_0 + d_1 + d_r < \omega - B, r \in \{3,4\} \end{equation}
Summing (\ref{ineq:6cyc4low3}) for $r \in \{3,4\}$, $2 \delta -d_3 -d_4 \leq d_5$, $2\delta - d_1 \leq d_2$ and $(\omega-2)(2 \delta - d_0) \leq (\omega-2)d_5$ gives:
\begin{equation} \label{ineq:6cyc4low4} (\omega - \frac{3}{2}) d_5 +  d_2 > -(\omega - B) +  ( \omega - \frac{1}{2}) (B-1) \end{equation}
If $\omega \leq \frac{9}{4}$: Multiplying (\ref{ineq:6cyc4low1a}) by $\omega - \frac{3}{2} > 0$, (\ref{ineq:6cyc4low4}) by $\omega -3\leq 0$ and summing them gives: $ (2\omega - \frac{9}{2}) d_2 < \frac{3}{2} (\omega -B)- \frac{3}{2}(B-1) $. Since $d_5 \leq \delta = \frac{B-1}{2} $ and $2 \omega -\frac{9}{2}\leq 0$, $(  2\omega - \frac{9}{2}) \frac{B-1}{2} < \frac{3}{2} (\omega -B)- \frac{3}{2}(B-1)$, which implies $B < \frac{10 \omega -3}{4 \omega +3}$.

Subtracting (\ref{ineq:6cyc4low1a}) by  (\ref{ineq:6cyc4low4}) gives: $(\frac{9}{2}  -2 \omega) d_5 <  2(\omega-B) - (2\omega - \frac{5}{2}) (B-1)$. Since $d_5 \geq (B-1) - d_0 \geq (B-1) -(2-B)$ and $\frac{9}{2}  -2 \omega \geq 0$, $(\frac{9}{2}  -2 \omega)((B-1) -(2-B)) <  2(\omega-B) - (2\omega - \frac{5}{2}) (B-1)$, which implies $B < \frac{22-4\omega}{17-4\omega}$

If $\omega > \frac{9}{4}$: Multiplying (\ref{ineq:6cyc4low1a}) by $2\omega^2 - 9\omega + 11> 0$, (\ref{ineq:6cyc4low1b}) by $4\omega -9 > 0$, (\ref{ineq:6cyc4low4}) by $2(\omega-2)(\omega -4)< 0$, and summing them gives:
$0< (7 \omega-14)(\omega -B) - 2(2\omega-1)(\omega-2) (B-1)$, which implies $B < \frac{11 \omega-2}{4 \omega +5}$

If $\omega \leq \frac{5}{2}$: Note that $d_2 \geq (B-1) - d_1 \geq (B-1) - (2-B) = 2B-3, d_5 \geq (B-1) -d_0 \geq (B-1) - (2-B) = 2B-3$. Summing (\ref{ineq:6cyc4low1a}) with $-d_2 \leq -(2B-3)$ and $(\omega-3) d_5 \leq (\omega-3) (2B-3)$ gives: $0 < \omega - B - (\omega-2) (B-1) - (4-\omega) (2B-3)$, which implies $B < \frac{10 - \omega }{7 -\omega}$
\item $P_{3,0} \leq 1 + d_4 + d_5 \leq 1+2\delta = B, P_{0,3} \leq 1+d_1+d_2 \leq 1+2\delta = B$, so $C_6 < \max\{P_{0,3}, P_{3,0}\} \leq B$.
\end{enumerate}
\end{document}